\documentclass[preprint,10pt,3p,authoryear]{elsarticle}

\let\today\relax
\makeatletter
\def\ps@pprintTitle{%
    \let\@oddhead\@empty
    \let\@evenhead\@empty
    \def\@oddfoot{\footnotesize\itshape
         {Submitted preprint} \hfill\today}%
    \let\@evenfoot\@oddfoot
    }
\makeatother

\usepackage{latexsym}
\usepackage{amsmath, amsfonts, amsthm,amssymb}
\usepackage{epsfig}
\usepackage{stmaryrd}
\usepackage{multicol}
\usepackage{pdfsync}
\usepackage{bm}
\usepackage{dsfont}
\usepackage{comment}
\usepackage{algorithmic}
\usepackage{algorithm}
\usepackage{caption}
\usepackage{psfrag}
\usepackage{subfig}
\usepackage{xcolor}
\usepackage{graphics}
\usepackage{enumitem}
\usepackage{mathtools}
\usepackage{bbm}
\usepackage{mathrsfs} 
\usepackage{color}
\usepackage{hyperref}
\usepackage{cancel}
\usepackage{float}
\usepackage{todonotes}
\usepackage{setspace}

\usepackage[utf8]{inputenc}
\usepackage{titlesec}
\titleformat{\subsection}
{\normalfont\fontsize{10}{17}\bfseries}{\thesubsection}{1em}{}
\allowdisplaybreaks 

\usepackage{tikz-cd}
\journal{TBA}
\hypersetup{
	colorlinks = true, 
	urlcolor = blue, 
	linkcolor = red, 
	citecolor = blue 
}

\newtheorem{theorem}{Theorem}

\newtheorem{lemma}{Lemma}
\newtheorem{proposition}{Proposition}

\theoremstyle{remark}

\numberwithin{equation}{section}

\mathtoolsset{showonlyrefs=true}

\titleformat{\paragraph}[runin]{\normalfont\normalsize}{\theparagraph}{1em}{}[\,\bfseries .]
\titlespacing*{\paragraph}{0pt}{3.25ex plus 1ex minus .2ex}{10pt}

\title{
\textbf{Strategic Learning and Trading in Broker-Mediated Markets}
}
\tnotetext[label0]{We thank S.~N.~Cohen and M.~Monoyios for comments on an earlier draft. We also thank participants at the OMI internal seminar,  the Stochastic Finance at Warwick seminar, and the London--Oxford--Warwick Financial Mathematics workshop.
}
\author[label1,label2]{Alif Aqsha}

\address[label1]{Mathematical Institute, University of Oxford}
\address[label2]{Oxford-Man Institute of Quantitative Finance}
\ead{Alif.Aqsha@maths.ox.ac.uk}
\author[label2]{Fayçal Drissi}
\ead{faycal.drissi@eng.ox.ac.uk}
\author[label1,label2]{Leandro S\'{a}nchez-Betancourt}
\ead{leandro.sanchez-betancourt@kcl.ac.uk}

\begin{document}
\newtheorem{rem}{Remark}
\newcommand{\mat}[1]{{\mathcal{M}_{#1}(\mathbb{R})}}

\newcommand{\la}{\left \langle}
\newcommand{\ra}{\right\rangle}
\newcommand{\cb}[1]{{\color{blue} #1}}
\newcommand{\norm}[1]{\left\lVert #1 \right\rVert}
\newcommand{\bae}{\begin{equation}\begin{aligned}}
\newcommand{\eae}{\end{aligned}\end{equation}}
\newcommand{\beq}{\begin{equation}}
\newcommand{\eeq}{\end{equation}}
\newcommand{\N}{\mathbb{N}}
\newcommand{\R}{\mathbb{R}}
\newcommand{\E}{\mathbb{E}}
\newcommand{\Pb}{\mathbb{P}}
\newcommand{\PbI}{\mathbb{P}^I}
\newcommand{\PbB}{\mathbb{P}^B}
\newcommand{\Lb}{\mathbb{L}}

\newcommand{\mfT}{{\mathfrak{T}}}
\newcommand{\mfO}{{\mathfrak{O}}}
\newcommand{\tT}{{t\in\mfT}}
\newcommand{\mcA}{{\mathcal{A}}}
\newcommand{\mcC}{{\mathcal{C}}}
\newcommand{\mcF}{{\mathcal{F}}}
\newcommand{\mcN}{{\mathcal{N}}}
\newcommand{\mcB}{{\mathcal{B}}}
\newcommand{\mcH}{{\mathcal{H}}}
\newcommand{\alt}{{\textup{alt}}}

\newcommand{\transB}{\mathfrak{t}}
\newcommand{\decayB}{\mathfrak{p}}
\newcommand{\instantB}{\mathfrak{h}}

\newcommand{\assign}{:=}
\newcommand{\nobracket}{}
\newcommand{\tmop}[1]{\ensuremath{\operatorname{#1}}}
\newcommand{\tmtextit}[1]{\text{{\itshape{#1}}}}
%

\newcommand{\F}{\mathcal{F}}
\newcommand{\Prob}{\mathbb{P}}
\newcommand{\X}{\mathbb{X}}
\newcommand{\Ss}{\mathcal{S}}

\newcommand{\Real}{\mathbb{R}}
\newcommand{\Y}{\mathbb{\mathcal{Y}}}
\newcommand{\tmL}{\mathbb{L}}
\newcommand{\Ll}{\mathcal{L}}
\newcommand{\Pp}{\mathcal{P}}
\newcommand{\Bb}{\mathcal{B}}
\newcommand{\Dd}{\mathcal{D}}
\newcommand{\Ee}{\mathcal{E}}
\newcommand{\Nn}{\mathcal{N}}
\newcommand{\Q}{\mathbb{Q}}
\newcommand{\ch}{\textbf{1}}

\newcommand{\tempB}{\mathfrak{a}}
\newcommand{\tempI}{\mathfrak{b}}
\newcommand{\tempU}{\mathfrak{c}}
\newcommand{\permB}{\mathfrak{p}}
\newcommand{\termpB}{\phi}
\newcommand{\termpI}{\psi}
\newcommand{\runnB}{r^B}
\newcommand{\runnI}{r^I}

\renewcommand{\d}{\mathrm{d}}

\newcommand{\stateB}{\mathfrak{y}}
\newcommand{\stateI}{\mathfrak{x}}

\newcommand{\nuI}{\eta}
\newcommand{\nuB}{\nu}
\newcommand{\nuU}{\xi}

\newcommand{\nuBstar}{\nu^*}
\newcommand{\nuIstar}{\eta^*}

\newcommand{\varI}{\mathbb{V}^I}
\newcommand{\varB}[1][B]{\mathbb{V}^{#1}}
\newcommand{\varBalt}[1][\text{alt}, B]{\mathbb{V}^{#1}}

\newcommand{\betaI}{\beta_0^I}
\newcommand{\betaaI}{\beta_1^I}
\newcommand{\rhoI}{\rho_0^I}
\newcommand{\rhooI}{\rho_1^I}

\newcommand{\betaB}{\beta_0^B}
\newcommand{\betaaB}{\beta_1^B}
\newcommand{\rhoB}{\rho_0^B}
\newcommand{\rhooB}{\rho_1^B}

\newcommand{\runcostI}{\rho_0^I + \rho_1^I\, \varI}
\newcommand{\runcostB}{\rho_0^B + \rho_1^B\, \varB}
\newcommand{\runcostBalt}{\rho_0^B + \rho_1^B\, \varBalt}

\newcommand{\termcostI}{\beta_0^I + \beta_1^I\, \varI}
\newcommand{\termcostB}{\beta_0^B + \beta_1^B\, \varB}
\newcommand{\termcostBalt}{\beta_0^B + \beta_1^B\, \varBalt}

\newcommand{\constrB}{\tempB - f_2^2\,\tempI}

\newcommand{\MI}[1][I]{M^{#1}}
\newcommand{\MB}[1][B]{M^{#1}}
\newcommand{\MBa}[1][B]{\tilde{N}^{#1}}
\newcommand{\MBb}[1][B]{\tilde{M}^{#1}}
\newcommand{\MZ}[1][Z]{M^{#1}}

\newcommand{\QI}[1][I]{Q^{#1}}
\newcommand{\QB}[1][B]{Q^{#1}}
\newcommand{\QU}[1][U]{Q^{#1}}
\newcommand{\XI}[1][I]{X^{#1}}
\newcommand{\XB}[1][B]{X^{#1}}
\newcommand{\XU}[1][U]{X^{#1}}

\newcommand{\QIstar}[1][I]{Q^{#1,*}}
\newcommand{\QBstar}[1][B]{Q^{#1, *}}
\newcommand{\XIstar}[1][I]{X^{#1, *}}
\newcommand{\XBstar}[1][B]{X^{#1, *}}

\newcommand{\mcFB}{\mcF^B}
\newcommand{\mcFI}{\mcF^I}
\newcommand{\vfB}{J^{B*}}
\newcommand{\vfI}{J^{I*}}
\newcommand{\pcB}{J^B}
\newcommand{\pcI}{J^I}
\newcommand{\mcAB}{\mcA^B}
\newcommand{\mcAI}{\mcA^I}

\newcommand{\g}{g}
\newcommand{\gI}{g^{I}}
\newcommand{\gB}{g^{B}}
\newcommand{\gZ}{g^{Z}}
\newcommand{\gY}{g^{Y}}
\newcommand{\h}{h}
\newcommand{\hI}{h^{I}}
\newcommand{\hB}{h^{B}}
\newcommand{\hZ}{h^{Z}}
\newcommand{\hY}{h^{Y}}
\newcommand{\f}{f}
\newcommand{\fI}{f^{I}}
\newcommand{\fB}{f^{B}}
\newcommand{\fZ}{f^{Z}}
\newcommand{\fY}{f^{Y}}
\newcommand{\GF}{{\mathfrak{G}}}
\newcommand{\kF}{{\mathfrak{K}}}
\newcommand{\hf}{{\mathfrak{h}}}

\newcommand{\Pn}{P^{\hat{\nu}}}
\newcommand{\Pa}{P^{\hat{\alpha}}}

\newcommand{\sigmaU}{\sigma^U}
\newcommand{\sigmaM}{\sigma^M}
\newcommand{\sigmaS}{\sigma^S}

\newcommand{\fay}[1]{{\todo[fancyline,backgroundcolor=red!20!white]{\begin{spacing}{0.5}{\tiny #1}\end{spacing}}}}
\newcommand{\alif}[1]{{\todo[fancyline,backgroundcolor=red!20!white]{\begin{spacing}{0.5}{\tiny #1}\end{spacing}}}}
\newcommand{\leandro}[1]{{\todo[fancyline,backgroundcolor=blue!20!white]{\begin{spacing}{0.5}{\tiny #1}\end{spacing}}}}

\newcommand\leo{ \color{red}}

\begin{abstract}
We study strategic interactions in a broker-mediated market in which agents learn and exploit each other’s private information. A broker provides liquidity to an informed trader and to noise traders while managing inventory in a lit market. The informed trader infers the broker’s trading activity in the lit market, while the broker estimates the trader’s private signal. Information leakage in the client’s trading flow generates economic value for the broker that is comparable in magnitude to transaction costs: the broker can speculate profitably and manage risk more effectively, which in turn adversely affects the informed trader’s performance. Brokers therefore hold a strategic advantage over traders who rely solely on prices to filter information. When the broker only relies on prices rather than client trading flow to infer information, their trading performance becomes indistinguishable from the performance of a naive strategy that internalises noise flow, externalises informed flow, and offloads inventory at a constant rate.\\

\end{abstract}

\maketitle
\section{Introduction}
We study a model in which a broker (“she”) provides liquidity to an informed trader (“he”) and to noise traders by acting as their counterparty. Simultaneously, the broker manages her inventory exposure by trading in the lit market, where she faces both transaction costs and price impact. In this setting, the strategic agents, i.e., the broker and the informed trader, operate under partial and heterogeneous information. The informed trader observes a private signal, and his trading strategy accounts for the price impact of the broker’s unobserved trading activity, which he estimates from market prices. Meanwhile, the broker manages her exposure in the lit market, and her strategy accounts for the unobserved trader’s private signal, which she estimates either from her clients’ trading flow or from market prices.

In this paper, we show how each agent estimates the information they do not observe by filtering relevant processes. Specifically, the informed trader relies solely on market prices, which in liquid markets have a low signal-to-noise ratio, resulting in imprecise estimates. We show that, in practice, filtering based only on prices does not yield accurate estimates of the broker’s activity. In contrast, the broker has access to higher-quality information through her clients’ trading flow, in addition to market prices. The trading flow of informed clients produces significantly more informative signals, which improves the broker’s trading performance and protects her against adverse selection. As a result, the broker holds a strategic advantage over an informed trader who trades exclusively with her. In practice, informed traders therefore need to conceal information or distribute their trading flow across multiple brokers to preserve privacy.\footnote{See \cite{oomen2017execution} for a study of broker aggregation.}

We employ stochastic control tools to derive the optimal adaptive strategies of both agents when prices are used as the information source. Both agents maximise terminal expected wealth subject to inventory penalties, with risk aversion scaled according to the conditional variance of their respective estimators. We characterise the optimal behaviour of both players and provide conditions for the existence and uniqueness of optimal Markovian strategies, which we derive explicitly. Owing to the linear–quadratic structure of the problem, we apply the separation principle and characterise the solutions as systems of linear and matrix Riccati differential equations (RDEs). The matrix RDE governing the broker’s problem is nonstandard. When the broker learns the signal from prices, we show that this RDE admits a unique solution, which in turn guarantees a unique optimal Markovian strategy. When the broker instead uses the informed trader’s order flow, the matrix RDE admits a unique solution up to a non-explosion condition.

We  illustrate the adaptive strategies of both agents in extensive numerical experiments. Our results confirm that market prices constitute a poor information source for filtering: the broker’s price impact is too small relative to price noise to produce reliable estimates of her own trading activity, and price volatility makes it difficult to infer the informed trader’s private signal accurately. Consequently, the broker’s performance in this case resembles that of a naïve strategy that internalises noise flow, externalises toxic flow, and unwinds inventory at a constant rate. 

In contrast, when the broker uses the informed trader’s order flow to estimate the signal, her performance improves substantially. This improvement arises for two reasons: higher-quality signals allow the broker to speculate more effectively on future prices, and they enable more precise internalisation and externalisation decisions. In particular, the broker increases externalisation and speculative trading when signal values are high or when her estimates are more precise. We show that the economic value of the improved filtering is of the same order of magnitude as transaction costs. Finally, we study the robustness of the best-performing estimator and find that it remains robust to inventory misspecification during the early part of the trading horizon.

\paragraph{\textbf{Literature review}} 
Our work combines elements from the execution literature for liquidity takers and the market-making literature for liquidity providers.\footnote{Early contributions to algorithmic trading include \cite{bertsimas1998optimal}, \cite{ho1981optimaldealer}, \cite{almgren2001optimal}, and \cite{avellanedastoikov2008}, which spurred a large body of subsequent research on execution and market making; see, for example, the textbooks \cite{cartea2015algorithmic} and \cite{gueant2016financial}. } Our work also relates to the literature on stochastic games. A closely related contribution is \cite{bank2021liquidity}, which develops a continuous-time trading game with two interacting markets: a dealer market in which agents trade at an equilibrium price and liquidity providers charge a fee to execute liquidity takers’ demands, and an open market in which the demands originating in the dealer market are absorbed at an unaffected martingale price.\footnote{A broader literature studies stochastic games in financial markets. Examples include portfolio liquidation games with self-exciting order flow \cite{guanxing_portfolio_2022}, games with price-impact interaction under open-loop, closed-loop, and centralised equilibria \cite{micheli2023closed}, and Nash competition among market makers for client order flow \cite{herdegen2023liquidity}. Related work analyses infinite-dimensional games with mean-field interactions \cite{jaber2023equilibrium}, as well as optimal liquidation problems with  feedback effects on future order flow \cite{chenoptimaltrade2023}.} In contrast, the flow of information plays a central role in our setting. Only the broker has price impact, as she is the sole agent trading in the lit market where transactions are anonymous, and this impact is unobserved by the informed trader. Conversely, the informed trader observes a private trading signal that is not available to the broker. As a result, the broker and the informed trader engage in filtering observed processes to infer the information they do not possess.\footnote{Our work also relates to  literature combining filtering techniques and stochastic optimal control. Early contributions include hedging (\cite{pham2001mean,Pham2001optimal}), insider trading models  (\cite{aase2012strategic}), and portfolio choice (\cite{Lim_Quenez_2015}). Filtering is also studied in mean-field games with heterogeneous agents (\cite{casgrain2018mean,firoozi2018mfg}), robust control (\cite{allan2019parameter,allan2020pathwise}),  market making  (\cite{campi2020,drissi2022solvability,cartea2024strategic}), and costly information acquisition (\cite{knochenhauer2024continuous}).
}

Closest to our work are \cite{cartea2022broker} and \cite{bergault2024mean}, both of which study strategic interactions between a broker and her clients. In \cite{cartea2022broker}, the broker faces both informed and uninformed traders, and both the broker and the informed trader optimise their objectives under ambiguity aversion to disentangle their strategic interactions.\footnote{Related strands of the literature include work on internalisation and externalisation decisions in FX markets \cite{butz2019internalisation,barzykin2023algorithmic,cartea2023detecting}, filtering of toxic flow \cite{barzykin2024unwinding}, and trading with predictive signals in execution problems \cite{cartea2016incorporating,neuman2022optimal,drissi2022solvability,bergault2022multi}. Recently, \cite{cartea2024nash} solve the perfect-information Nash equilibrium of the game. In contrast, we do not assume perfect information; instead, both agents must filter and infer the information they do not observe.
} On the other hand, \cite{bergault2024mean} study a large population of  traders with both private and common signals, and characterise the Nash equilibrium in the mean-field limit. We depart from this approach by modelling asymmetric information acquisition. In our framework, both agents actively filter the information they do not observe within their respective decision problems. To the best of our knowledge, this paper, together with \cite{wu2024broker}, is among the first to study settings in which both liquidity providers and  takers optimise their strategies while filtering missing information. While \cite{wu2024broker} analyse the Nash equilibrium in a multidimensional setting, our approach adopts a Stackelberg-type formulation.

The remainder of the paper proceeds as follows. Section \ref{sec: model} reviews the key elements of the model. Section \ref{sec: informed trader} formulates the filtering problem of the informed trader and solves his stochastic control problem. Section \ref{sec: broker} introduces the filtering problems of the broker and solves both of her stochastic control problems. 
In particular, we 
propose two alternative filters (and solve their respective control problems) for how  the broker extracts an estimate of the signal that she does not observe. 
Section \ref{sec:numerical_results} carries out simulations to study the performance of the strategies we derive against three  benchmarks.

\section{Modelling framework}\label{sec: model}

Let $T>0$ be a fixed trading horizon and let $\mfT = [0,T]$. We fix a filtered probability space $\left(\Omega, \mcF, (\mcF_t)_{t\in\mfT}, \Pb \right)$ satisfying the usual conditions of right continuity and completeness. The probability space supports four standard Brownian motions $W^S$, $W^\alpha$, $W^B$, and $W^U$. 
The full-information filtration $(\mcF_t)_{t\in\mfT}$ is that generated by the four Brownian motions. 
We consider two sub-filtrations in our model: (i) the filtration $\left(\mcFI_t\right)_{t\in\mfT}$ of the informed trader, and (ii) the filtration $\left(\mcFB_t\right)_{t\in\mfT}$ of the broker, which we define below. 

The broker provides liquidity to the informed trader and to the uninformed trader by acting as counterparty to their trades throughout the trading window $\mfT$. The informed trader and the uninformed trader carry out their trades at speeds $\nuI$ and $\nuU$ respectively. Positive values of the trading speed denote buying, and negative values  denote selling.   As remuneration for this service, the broker charges transaction costs to both traders. In our model, the transaction costs charged by the broker are linear in her clients' trading speed and are modelled by the positive coefficients  $\tempI$ (for the informed trader), and $\tempU$ (for the uninformed trader).\footnote{If the broker charges the same fee to all clients, then we take $\tempI=\tempU$.} 
In parallel to providing liquidity, the broker manages her exposure and inventory  
by trading in the open (or lit) market at a trading speed $\nuB$. We assume that the broker incurs linear transaction costs modelled by the positive constant $\tempB$, and that her trading activity in the market generates a linear permanent impact modelled by the coefficient $\permB > 0$.

The midprice of the asset is driven by a stochastic signal $\alpha$ which is privately observed by the informed trader. Let $\left(S_t\right)_{t\in\mfT}$ denote the process describing the midprice of the asset, whose  dynamics are
\begin{equation}
\label{eq: price dynamics}
  \d S_t = (\permB\,\nuB_t + \alpha_t)\,\d t + \sigma^S \d W_t^S\,,\qquad S_0 \in\mathbb{R}^+\,,
\end{equation}
where  $\sigma^S>0$ is the price volatility, and $W^S$ is a standard Brownian motion. The signal process $\left(\alpha_t\right)_{t\in\mfT}$  follows the Ornstein–Uhlenbeck dynamics 
\begin{equation}
 \d \alpha_t = -\kappa^\alpha\,\alpha_t\,\d t + \sigma^\alpha\,\d W^\alpha_t\,,\qquad \alpha_0 \in\mathbb{R}\,,
\end{equation}
where $\kappa^\alpha\geq 0$ is the mean-reversion parameter, $\sigma^\alpha>0$ is the volatility of the signal, and $W^\alpha$ is a Brownian motion. The signal and price processes may share common stochastic factors, so  we assume that $W^\alpha$ and $W^S$ are correlated and write $\langle W^S, W^\alpha \rangle_t  = \rho\,t$, where $\rho \in [-1, 1]$.

Let $\left(\QB_t\right)_{t\in\mfT}$ and $\left(\XB_t\right)_{t\in\mfT}$ be the processes that describe 
the inventory  and the cash of the broker, respectively, with dynamics
\begin{align}
    \d \QB_t &= \left(\nuB_t - \nuI_t - \nuU_t\right)\d t\,,\qquad Q^B_0 = q^B\in\mathbb{R}\,,\label{eq: inventory of broker}\\
    \d \XB_t &= - \left(S_t + \tempB\,\nuB_t\right)\nuB_t\,\d t +   \left(S_t + \tempI\,\nuI_t\right)\nuI_t\,\d t +   \left(S_t + \tempU\,\nuU_t\right)\nuU_t\,\d t\,,\qquad X^B_0 = 0\,,\label{eq: cash of broker}
\end{align}
where $\tempB,\tempI,\tempU$ are the transaction costs coefficients described above. Similarly, the informed trader's inventory process $\left(\QI_t\right)_{t\in\mfT}$ and cash process $\left(\XI_t\right)_{t\in\mfT}$ follow the dynamics
\begin{align}
    \d \QI_t &= \nuI_t\,\d t\,,\qquad Q^I_0 = q^I\in\mathbb{R}\,,\label{eq: inventory of informed}\\
    \d \XI_t &= - \left(S_t + \tempI\,\nuI_t\right)\nuI_t\,\d t\,,\qquad X^I_0 = 0\,.\label{eq: cash of informed}
\end{align}
Below, we use the notation $P^{I,\nuI}$ to emphasise that a process $P$ describes a property of the informed trader and is controlled through the process $\nuI.$ Similarly, we use the notation $P^{B,\nuB}$ for the broker.

In our model, the broker manages her inventory  while providing liquidity to her clients, and the informed trader exploits his private signal. 
Both agents assume parametric models for their competitor's behaviour, and employ the information they observe to filter their competitor's private information. 
On the one hand, the broker observes the midprice $S$ of the asset, and given that she is counterparty to the trades of her clients, she also observes the trading speed of the informed trader $\nuI$ and the trading speed of the uniformed trader $\nuU$. In the main model, the broker filters the midprice minus her impact on prices to estimate the private signal $\alpha$ (which 
she does not observe). 
On the other hand, the informed trader observes his signal $\alpha$ and the midprice $S$, but he does not observe the trading speed of the broker. The informed trader assumes a parametric form for the broker's strategy, which he uses to estimate the broker's trading  speed $\nuB$. Both agents employ the Kalman-Bucy filter in their respective control problems. Sections \ref{sec: informed trader} and \ref{sec: broker} derive the optimal strategies of both agents.

\section{The informed trader}\label{sec: informed trader}

This section derives the optimal adaptive strategy of the informed trader. 

\subsection{Learning equations} In our model, the informed trader does not directly observe the trading activity of the broker in the open market, that is, he does not observe $\nuB$. However, he observes the signal $\alpha$ and he knows that the broker's trading activity has a permanent impact. Thus, to infer the broker's trading rate, he uses the 
process $\left(Y_t\right)_{t\in \mfT}$ with dynamics
\begin{equation}
  \d Y_t = \d S_t - \alpha_t \, \d t = \permB\,\nuB_t \, \d t + \sigma^S \d W_t^S\,.
\end{equation}

In practice, the informed trader has little a-priori knowledge on the behaviour of the broker. Here, he considers a probability measure $\PbI$ where the  dynamics of the price follow \eqref{eq: price dynamics} but with the following parametric form for $\nuB$, given by
\begin{equation}
\label{eq: nuB_according_to_trader}
\d \nuB_t = -\theta^B\, \nuB_t \, \d t + \sigma^B \d W_t^B \,,
\end{equation}
where $\theta^B, \sigma^B > 0$, and $W^B$ is a Brownian noise independent of $W^S$. The dynamics in \eqref{eq: nuB_according_to_trader} are a key modelling assumption to derive the closed-form solution to the informed trader's adaptive control problem, i.e., with filtering. The parameters $\theta^B$ and $\sigma^B$ encode what the informed trader believes about the dynamics of the trading speed of the broker.

Here, we use the Kalman-Bucy filter for processes with linear drift and diffusion terms (see e.g., \cite{filtering_alain_crisan} and Chapter 22 in \cite{Cohen15}) to solve the informed trader's stochastic filtering problem of $\nuB$ from observing $Y$. The filter is characterised by the dynamics of its mean and (co)variance. More precisely, assume \eqref{eq: nuB_according_to_trader} holds and let $\mcFI_t = \sigma(\alpha_s, Y_s:\, s\in [0,t])\vee \mcN$  where $\mcN$ is the set of $\Pb$-null sets of the full sigma-algebra $\mcF$,  $\mcFI = \sigma\left( \bigcup_{t\in\mfT} \mcFI_t\right)$, and $\mathcal{C} \vee \mathcal{D}$ denotes the smallest sigma-algebra containing $\mathcal{C}$ and $\mathcal{D}$.  Let $\hat{\nuB}_t = \E \left[ \nuB_t\, |\, \mcFI_t\right]$ be the projection of $\nuB_t$ onto $\mcFI_t$ and let $\varI_t = \E \left[ (\nuB_t - \hat{\nuB}_t)^2\, |\, \mcFI_t \right]$ be its conditional variance. The \cite{kalman1961new} filter gives 
\begin{align}
  \d \hat{\nuB}_t &= -\theta^B \,\hat{\nuB}_t \, \d t + \frac{1}{(\sigma^S)^2}\, \permB\, \varI_t\,  \left( \d Y_t - \permB\, \hat{\nuB}_t \, \d t \right) \,,\label{eq: hat nuB} \\
  \d \varI_t &= \left[(\sigma^B)^2 - 2\,\theta^B\, \varI_t - \frac{\permB^2\, \left(\varI_t\right)^2}{(\sigma^S)^2}\right] \, \d t \,. \label{eq: varI}
\end{align}

For an initial distribution of
$\nuB_0$, the informed trader solves \eqref{eq: hat nuB} to obtain the estimate $\hat{\nuB}$ of 
$\nuB$, with confidence intervals given by the variance $\varI$ in \eqref{eq: varI}. In particular, the dynamics of the variance in \eqref{eq: varI} are described by an ordinary differential equation (ODE) of the Riccati-type that admits a closed-form solution. 
For simplicity, in what follows, the informed trader assumes $\nuB_0=0$ so he sets $\varI_0 = 0$.

Equipped with the estimate  $\hat{\nuB}$ of the broker's strategy in \eqref{eq: hat nuB},   the informed trader derives that under $(\Omega,\mcFI, (\mcFI_t)_{t\in\mfT}, \PbI)$, the midprice $S$ follows the dynamics 
\begin{align}
  \d S_t = & (\permB\,\hat{\nuB}_t + \alpha_t) \, \d t + \sigma^S\,\d \tilde{W}_t^{S,I}\,,
\end{align}
where $\hat{\nuB}$ is as in \eqref{eq: hat nuB} and $\tilde{W}^{S,I}_t$ is a  $(\PbI, (\mcFI_t)_{t\in\mfT})$-Brownian motion.

\begin{rem}
    Note that the filtration $(\mcFI_t)_{t\in \mfT}$ is now fixed. This is because the latent $\nuB$ is assumed to be exogenous and does not depend on the informed trader's control. As such, the observation process $Y$ is already decoupled from the informed trader's control and there is no ``chicken and egg'' problem as in Section 9.1 of \cite{bensoussan2018estimation}.
\end{rem}

\subsection{The control problem of the informed trader} 
The informed trader exploits her trading signal to speculate and maximise her trading performance. However, he is averse to uncertainty about his competitor's trading behaviour.  More precisely, the set of  admissible strategies is
\begin{equation}
    \mcAI := \left\{ \nuI = \left(\nuI_t\right)_\tT \, :\, \nuI \text{ is $(\mcFI_t)_{t\in\mfT}-$progressively measurable } \text{and} \, \, \E^I \left[ \int_0^T (\nuI_s)^2\, ds \right] < \infty \right\}\,,
\end{equation}
and for every $\nuI\in\mcAI$, 
the informed trader maximises the performance criterion
\begin{equation}\label{eq: perf crit informed}
    H^{I,\nuI}(t, \stateI) = \E^I_{t, \stateI} \left[ \XI[I,\nuI]_T + \QI[I,\nuI]_T\,S_T  - (\beta_0 + \beta_1\, \varI_T)\,\left(\QI[I,\nuI]_T\right)^2 - \int_t^T (\rho_0 + \rho_1\, \varI_s) \left(\QI[I,\nuI]_s\right)^2 \, \d s \right] \,,
\end{equation}
where $\beta_0,\beta_1,\rho_0, \rho_1>0$, $\stateI = (s, x^I, q^I, \hat{\nuB})$, and the expectation is under $\PbI$. 
The first two terms in the expectation in \eqref{eq: perf crit informed} are the informed trader's terminal wealth composed of the marked-to-market value of his inventory and his terminal cash. The last two terms represent a terminal and running risk aversion components, respectively. In particular, we introduce the new terms $\beta_1\,\varI$ and $\rho_1\,\varI$ to adjust risk 
aversion according to the conditional variance $\varI_t$.\footnote{When $\beta_1 = \rho_1 = 0$ we recover a performance criterion that is widely used in algorithmic trading; see e.g., \cite{donnelly2022optimal} for a recent review article.} These terms penalise any terminal or running inventory according to the informed trader's uncertainty because higher conditional
variance implies more uncertainty in the estimation of the true $\nuB$. 
The associated control problem is 
\begin{align}
H^I(t, \stateI) = \sup_{\nuI\in\mcAI}  H^{I,\nuI}(t, \stateI) \,. \label{eq:trader_control_problem}
\end{align}

\subsection{The optimal adaptive strategy}
The above control problem has a linear-quadratic structure, hence, we are able to characterise the solution with a Riccati differential equation and linear ODEs. We show existence and uniqueness of these differential equations and use them to solve the control problem of the informed trader. 
The following theorem, whose proof is in the appendix,  shows the solution to the adaptive control problem of the informed trader.
\begin{theorem}\label{thm: solution Informed problem}
The HJB equation \eqref{eq:trader_hjb} admits a unique solution in $C^{1,2}([0,T], \Real^{|\stateI|})$ in the form of
\begin{equation}
\label{eq:trader_hjb_solution}
    H^I(t, \stateI) = x^I + q^I\, s + g_0^I  (t, \alpha, \hat{\nuB}) \, + q^I\, g_1^I  (t, \alpha, \hat{\nuB}) + \left(q^I\right)^2 \, g_2^I (t)\,,
\end{equation}
where $g_i^I$, $i=0,1,2$ are defined in the appendix (Lemma \ref{lemma:trader_DEs_existence_solutions}). Moreover, there exists a constant $c<\infty$ such that
\begin{equation}
    |H^I(t, \stateI)| \leq c\, \left(1 + \|\stateI\|^2 \right)\quad \forall (t,\stateI) \in [0,T]\times \Real^{|\stateI|}\,.
\end{equation}
Finally, the trader's Markovian optimal trading rate for the control problem \eqref{eq:trader_control_problem} in feedback form is
\begin{equation}
\label{eq:01_informed_strategy_feedback_form}
    \nuIstar_t = \frac{1}{2\,\tempI} \Big({z_1^I(t)\,\alpha_t + z_2^I(t)\,\hat{\nuB}_t  + 2\,g_2 ^I(t)\,\QI[I*]_t }\Big)\,.
\end{equation}
where $z^I_1, z^I_2$ are deterministic functions defined in the appendix and they satisfy that $z_1^I, z_2^I\geq 0$ and $\g_2^I < 0$.
\end{theorem}

\begin{rem}
If the informed trader were to observe $\nuB$ and if $\nuB$ was indeed an OU process, then the HJB equation of such full-information control problem would differ from \eqref{eq:trader_hjb} in the coefficients of $\partial_{\hat{\nuB},\hat{\nuB}} H^I$, $\partial_{s,\hat{\nuB}} H^I$, and $\partial_{\alpha,\hat{\nuB}} H^I$ only. From here, we use a version of Lemma \ref{lemma:trader_DEs_existence_solutions} to show that the equivalent coefficients to $g_2^I$ and $g_1^I$ we obtain in such full-information case are exactly the same but changing $\hat{\nuB}$ for $\nuB$. 
This is the separation principle discussed in Section 9.3 of \cite{bensoussan2018estimation}.
\end{rem}

The optimal strategy \eqref{eq:01_informed_strategy_feedback_form} of the informed trader has three components. The first component is a speculative component that exploits the private signal. The second component uses the estimate that he has about the trading speed of the broker to speculate about future price movements due to the permanent price impact of the broker. The last component reduces the inventory of the informed trader according to his uncertainty and to the level of transaction costs charged by the broker.

\section{The broker}\label{sec: broker}
This section derives the optimal adaptive strategy of the broker.

\subsection{Learning equations}\label{sec:learning equations broker}
The broker assumes a model where the informed trader's trading flow  is driven by his private signal. More precisely, 
she assumes that the informed trader executes the strategy \eqref{eq:01_informed_strategy_feedback_form}, so she assumes the parametric form 
\begin{equation}\label{eq: assumption broker}
    2\,\tempI\,\nuIstar_t = z_1^I(t)\,\alpha_t + z_2^I(t)\,\hat\nuB_t  + 2\,g_2^I (t)\,\QI[I*]_t\,,
\end{equation}
for deterministic functions $z_i(t)$ for $i=0,1,2$. 

However, the broker does not observe $\alpha$, nor can she construct the filter $\hat\nuB_t$ because the innovation process requires knowledge of $\alpha$. 
Thus,  the broker replaces the signal $\alpha$  in  \eqref{eq: assumption broker} with an estimate $\hat\alpha$. Below, Section \ref{sec: main model} assumes that the broker employs midprices to construct this estimate and in Section \ref{sec: alternatives} the broker employs the trading speed of the informed trader as the source of information. One of our main results follows from comparing the performance of these two approaches.

\subsubsection{Learning from prices}\label{sec: main model}

In this subsection the broker filters the signal from observing the price process. To do this, the broker 
replaces the unknown processes $\alpha$ and $\hat{\nuB}$ (used in the informed trader's strategy) with the known measurable processes $\hat{\alpha}$ and $\nuB$, respectively. More precisely, the broker assumes that her competitor's trading speed takes the form
\begin{equation}
\label{eq: broker_model_assumption}
    \nuIstar_t = \frac{1}{2\,\tempI} \Big( {z_1^I(t)\,\hat{\alpha}_t + z_2^I(t)\,\nuB_t  + 2\,g_2^I (t)\,\QI[I*]_t } \Big) = f_1(t)\,\hat{\alpha}_t + f_2(t)\,\nuB_t + f_3(t) \,\QI[I*]_t \,.
\end{equation}
This latter assumption may be conservative, as perfect knowledge gives an obvious advantage to the informed trader. However, we show in the numerical examples of Section \ref{sec:numerical_results} that the broker can strategically influence (manipulate) the informed trader through this component, i.e., through knowing and exploiting that the informed trader's strategy is based on that of the broker. The functions $f_i$ for $i=1,2,3$ are those induced by the second equality.

The broker observes the process
\begin{equation}\label{eq: broker_model_assumption price}
    \d Z_t = \d S_t  - \permB\,\nuB_t\, \d t = \alpha_t\, \d t + \sigma^S\, \d W^S_t\,.
\end{equation} 
In \eqref{eq: broker_model_assumption price}, there is an implicit assumption that the trading speed of the broker is no longer an exogenously  observed process, but instead it satisfies an equation that is linear  in the filter $\hat{\alpha}$, the broker's control $\nu$, and the integral of $\nuIstar_t$ itself. 

The broker  projects $\alpha$ onto the  filtration $\mathcal{Z}_t = \sigma(Z_s:\, s\in [0,t])\vee \mcN$, where $\mcN$ is the set of $\Pb$-null sets of the full sigma-algebra $\mcF$.
Let $\hat{\alpha}_t = \E \left[ \alpha_t\, |\, \mathcal{Z}_t\right]$ be the broker's estimate and let $\varB = \E \left[ (\alpha_t - \hat{\alpha}_t)^2\, |\, \mathcal{Z}_t \right]$ be its conditional variance. The Kalman Bucy filter gives
\begin{align}
\label{eq: filter eq1 broker}
\d \hat{\alpha}_t &= -\kappa^\alpha\,\hat{\alpha}_t\,\d t + \frac{1}{\left(\sigma^S\right)^2}\left( \varB_t  + \rho\,\sigma^S\,\sigma^\alpha \right)\, \left[ \d Z_t - \hat{\alpha}_t\, \d t \right]\,,
\end{align}
and 
\begin{align}
    \d \varB_t &= \left[ (1-\rho^2)\, \left(\sigma^\alpha\right)^2 + 2\, \left(-\kappa^\alpha - \frac{1}{\sigma^S}\, \rho\, \sigma^\alpha\right)\, \varB_t - \frac{1}{\left(\sigma^S\right)^2}\, \left( \varB_t \right)^2  \right]\, \d t \label{eq: cond var filter 1 broker}\,.
\end{align}
In particular, the conditional variance $\varB$  of the estimator $\hat{\alpha}$ is a deterministic function of time. Finally, the broker fixes a probability measure $\PbB$ such that under $(\Omega,\mcFB, (\mcFB_t)_{t\in\mfT}, \PbB)$, the price follows the dynamics
\begin{align}
	\d S_t &= (\permB\,\nuB_t + \hat{\alpha}_t)\,\d t + \sigma^S\, \d \tilde{W}_t^{S,B}\,,
\end{align}
where $\tilde{W}_t^{S,B}$ is a $(\PbB, (\mcFB_t)_{t\in\mfT})$-Brownian motion, and the filtration $\mcFB_t  = \sigma(S_s, Z_s, \nuU_s:\, s\in [0,t])\vee \mcN$.

\subsubsection{Learning from  trading flow}\label{sec: alternatives}
In this subsection  the broker employs the trading flow  as a source of information from which she obtains an estimate of the private signal $\alpha$. Observe that  $\nuIstar$ is a measurable stochastic process that carries information about the unknown signal $\alpha$; c.f. with the assumptions in Section \ref{sec: main model}.
Using the parametric form \eqref{eq: assumption broker}, it follows that the broker observes the process ${\gamma}_t:=\nuIstar_t - f_3(t)\, \QI[I*]_t = f_1(t)\, \alpha_t + f_2(t)\, \hat{\nuB}_t$. Observe that this implicitly assumes that the broker knows the value of the parameters in the dynamics of the alpha signal. Below, in the numerical section we carry out robustness checks to test how the performance of the broker's strategy changes as the broker mispecifies the values of these model parameters.

Using $\gamma_t$ we define a scaled version $\tilde{Z}^\alt_t$ and a shifted version $Z^\alt_t$. For the sake of presentation we report the details in the appendix. The key result is that one is left with the 
filtering framework
\begin{align}
    \d \alpha_t &= -\kappa^\alpha\, \alpha_t\, \d t + \sigma^\alpha\, \d W_t^\alpha\,,\\
    \d Z_t &= \mathfrak{G}_7(t)\, \alpha_t\, \d t + \d \tilde{W}_t\,,
\end{align}
where $Z^\alt$ is the observed process, $\mathfrak{G}_7$ is a deterministic function reported in  \ref{app: learning from order flow}, $\tilde W$ is a $(\Pb, (\F_t)_{t\in \mfT}$-Brownian motion, and $\alpha$ is the unobserved process.
Finally, she takes $(\mathcal{Z}_t)_\tT$ as the augmented filtration generated by the process $Z^\alt$. The proposition below computes, explicitly, the dynamics of the new filter.
\begin{proposition}
\label{prop: alternative filter from trader's speed dynamics}
    Let $\hat{\alpha}^\alt_t:=\E \left[\alpha_t \, |\, \mathcal{Z}_t\right]$ and $\varB[\textup{alt}, B]_t:=\E \left[(\alpha_t - \hat{\alpha}^\alt_t)^2 \, |\, \mathcal{Z}_t\right]$. The processes $\hat{\alpha}^\alt$ and $\varB[\textup{alt}, B]$ satisfy the following stochastic differential equations,
    \begin{align}
        \d \hat{\alpha}^\alt_t &= -\kappa^\alpha\, \hat{\alpha}^\alt_t\, \d t + \left[ \mathfrak{G}_7(t)\, \varB[\textup{alt}, B]_t + \sigma^\alpha\, \mathfrak{K}(t) \right]\, \d I_t\label{eq: alpha mean estimate}\,,\\
        \d \varB[\textup{alt}, B]_t &= \left[ (\sigma^\alpha)^2 - 2\, \kappa^\alpha\, \varB[\textup{alt}, B]_t - (\mathfrak{G}_7(t)\, \varB[\textup{alt}, B]_t + \sigma^\alpha\, \mathfrak{K}(t))^2 \right]\, \d t \label{eq: alpha var estimate}\,,
    \end{align}
    where $\mathfrak{K}(t)$ is a deterministic function and $I$ is an innovation process satisfying the equation
    \begin{equation}
        I_t := Z_t - \int_0^t \E \left[ \mathfrak{G}_7(s)\, \alpha_s \, \Big| \,\mathcal{Z}_s \right]\, \d s\,,
    \end{equation}
    which is a $(\Pb, (\mathcal{Z}_t)_{t\in\mfT})$-Brownian motion.
\end{proposition}
The proof of Proposition \ref{prop: alternative filter from trader's speed dynamics} is in \ref{app:proof of proposition}. As the noise flow $\nuU$ is independent of either $W^S$ and $W^{\alpha}$, the above result would not change if we replace the filtration $(\mathcal{Z}_{t})_{t\in \mfT}$ with the filtration $\mcF^{B,\alt}_t  := \sigma(Z^\alt_s, \nuU_s:\, s\in [0,t])\vee \mcN$.

\subsection{The control problems of the broker}  
This section derives the closed-form optimal strategy  of the broker. Section \ref{solving when broker filters prices} employs the learning equations of Section \ref{sec: main model}, i.e., when she employs the midprice as the source of information, and Section \ref{solving when broker filters flows} employs the learning equations of Section \ref{sec: alternatives}, i.e., when she employs the trading speed of the informed trader as the source of information. For the first, the filtration of the broker is $(\mcFB_t)_{t\in \mfT}$, and for the second, the filtration is $(\F^{B,\alt}_t)_{t\in \mfT}$.

In the numerical experiments below, we study the performance of the broker's strategy when she estimates the private signal according to \eqref{eq: filter eq1 broker}--\eqref{eq: cond var filter 1 broker} and also \eqref{eq: alpha mean estimate}--\eqref{eq: alpha var estimate}.

The uninformed order flow $\nuU$ is  mean-reverting around zero and we write its dynamics as
\begin{align}
\d\nuU_t &= -\kappa^U\, \nuU_t\,\d t + \sigma^U\, \d W_t^U\,, \quad \nuU_0 = 0\,.
\end{align}
The inventory $\QB$ of the broker evolves with the orders that she fills for her clients and her own trading in the lit market as follows
\begin{align}
\d \QB[B,\nuB] &= \left(\nuB_t - \nuIstar_t - \nuU_t\right)\, \d t\,, \quad \QB[B,\nuB]_0=0\,,\\
\d \QI[I,\nuIstar]_t &= \nuIstar_t\, \d t\,,\quad \QI[I,\nuIstar]_0=0\,.
\end{align}
Similarly, the broker earns the transaction costs charged to her clients and incurs transaction costs from her own trading activity, so her cash $\XB$ 
follows the dynamics
\begin{align}
\d \XB[B,\nuB]_t &= \big( -\nuB_t\, \left(  S_t + \tempB\, \nuB_t\right) + \nuIstar_t\, \left(  S_t + \tempI\, \nuIstar_t\right) + \nuU_t\, \left(  S_t + \tempU\, \nuU_t\right)\big)\, \d t\,, \quad \XB[B,\nuB]_0=0\,,
\end{align}
where the transaction cost coefficients $\{\tempB,\tempI,\tempU\}$ are described in Section \ref{sec: model}.

For a given strategy $\nuB$, she maximises the performance criterion 
\begin{align}
&H^{B,\nuB}(t,\stateB)= \E_{t,\stateB} \left[ \XB[B,\nuB]_T + \QB[B,\nuB]_T \,S_T  - \left(\beta_0^B +
		\beta_1^B\,\varB_T\right) \left(\QB[B,\nuB]_T\right)^2 - \int_t^T \left(\rho_0^B + \rho_1^B\,\varB_s\right) \left(\QB[B,\nuB]_s\right)^2 \, \d s \right]
\end{align}
where $\stateB = (s, x^B, q^B, \hat{\alpha}, \nuU, q^I)$ for when the broker uses prices as the source of information, and $\stateB = (s, x^B, q^B, q^I, \gamma, \hat{\alpha}, \nuU)$ when the broker uses the trading speed of the informed trader as the source of information.  Thus, the control problem that the broker  solves is 
\begin{equation}
\label{eq: broker control problem}
H(t,\stateB) = \sup_{\nuB\in\mcAB} H^{B,\nuB}(t,\stateB)\,.
\end{equation}
where the set of admissible control $\mcA^B$ depends on the broker's learning strategies. For the control problem when the broker filters prices, she considers the following set of admissible controls
\begin{equation}
    \mcA^{B} := \left\{ \nuB = \left(\nuB_t\right)_\tT \, :\, \nuB \text{ is $(\mcFB_t)_{t\in\mfT}-$progressively measurable } \text{and} \, \, \E^B \left[ \int_0^T (\nuB_s)^2\, ds \right] < \infty \right\}\,.
\end{equation}
When the broker filters the informed trader's flow, she considers the following set of admissible controls
\begin{equation}
    \mcA^{B} := \left\{ \nuB = \left(\nuB_t\right)_\tT \, :\, \nuB \text{ is $(\F^{B,\alt}_t)_{t\in\mfT}-$progressively measurable } \text{and} \, \, \E \left[ \int_0^T (\nuB_s)^2\, ds \right] < \infty \right\}\,.
\end{equation}

Here we solve the optimal control for both learning equations.

\subsubsection{Solving the control problem when the broker filters prices}
\label{solving when broker filters prices}
In  \ref{app: derivations broker} we provide detailed derivations of the equations that characterise the value function $H$, which is given by
\begin{equation}
  H(t,\stateB) = x + q^B\, s + G_0 + 2\, G_1\, \tilde{\stateB}^\top + \tilde{\stateB}\, G_2\, \tilde{\stateB}^\top
\end{equation}
where $\tilde{\stateB} = \begin{pmatrix} q^B & \hat{\alpha} & \nuU & q^I \end{pmatrix}$, and $G_0:[0,T]\to \Real$, $G_1:[0,T]\to \mat{1\times 4}$, and $G_2:[0,T]\to \mat{4\times 4}$  are deterministic functions of time reported in the appendix. The notation $\mat{n\times m}$ denotes the space of $n$ by $m$ matrices with values in $\mathbb{R}$. This ansatz helps us characterise the solution to the control problem provided that the Riccati equation that $G_2$ satisfies admits a solution. It turns out that there is a region in the parameter space for which this equation admits a unique solution.\\

The next lemma shows existence and uniqueness of a solution to the differential equation in $G_2$. We assume that the values of the model parameters are such that $1 + \tempI\, f_3(t) > 0$ for all $t\in [0,T]$. Note that  $\tempI\, f_3/2 = g_2^I$, thus, the inequality can be written as $g_2^I>-1/2$. Given $g_2^I$ satisfies the ODE in \eqref{eq:trader_sys_1}, we see that if the value of the temporary impact and terminal inventory risk parameters for the informed trader ($\tempI$, $\betaI$, and $\betaaI$) are small, this condition is met. In the numerical examples below the values of the model parameters are such that this condition is satisfied.
\begin{proposition}
\label{prop: freiling riccati existence}
    Assume that the values of model parameters are such that  $1 + \tempI\, f_3(t) > 0$ for all $t\in [0,T]$. 
    There exists $\delta>0$ such that for all $\permB\in [0, \delta)$, the matrix Riccati equation \eqref{eq: broker riccati} admits a unique solution $G_2: [0,T] \to \mat{4\times 4}$ with terminal condition \eqref{eq:term cond G2}. 
\end{proposition}

Once the existence and uniqueness of $G_2$ is established, the existence and uniqueness of $G_0$ and $G_1$ is straightforward. The following theorem, whose proof is in \ref{app: proof of theorem broker}, solves the adaptive control problem of the broker.

\begin{theorem}\label{thm: solution problem broker}
Let $\permB\in[0,\delta)$ with $\delta$ as in Proposition \ref{prop: freiling riccati existence}.
Then,  the HJB equation \eqref{eq: broker HJB} admits a unique solution in $C^{1,2}([0,T], \Real^{|\stateB|})$ with quadratic growth in $\stateB$. Moreover, the broker's Markovian optimal trading rate for the control problem \eqref{eq: broker control problem} in feedback form is given by 
\begin{equation}  \label{eq:broker_optimal_trading_rate_feedback}
    \nuBstar_t = \frac{1}{\sqrt{\tempB - \tempI\, f_2^2(t)}} \, \left( P_7 + 2\, P_8\, G_2(t) \right)\, \tilde{\stateB}^\top_t\,,
\end{equation}
where $\tilde{\stateB}_t = \begin{pmatrix} Q^{B*}_t & \hat{\alpha}_t & \nuU_t & Q^{I}_t \end{pmatrix}$, and $P_7,P_8\in\mat{1\times 4}$ are in \eqref{eq: P7} and \eqref{eq: P8} in the Appendix.
\end{theorem}

The optimal adaptive strategy \eqref{eq:broker_optimal_trading_rate_feedback} of the broker balances internalisation and externalisation of the trading flow of her clients and has four components. The first component reduces the inventory of the broker according to her model parameters---trading costs, trading revenue from clients, and aversion parameters. The second component exploits the estimated private signal of the informed trader. The third component adjust the strategy according to the trading activity of the uninformed traders. Finally, the last component adjusts the strategy according to the informed trader inventory level. Below, Section \ref{sec:numerical_results} thoroughly studies the internalisation and externalisation components in the strategy of the broker.

\subsubsection{Solving the control problem when the broker filters the trading speed}

\label{solving when broker filters flows}

In  \ref{app: derivations broker} we provide detailed derivations of the equations that characterise the value function $H$, which is given by
\begin{equation}
    H(t, \tilde{\stateB}) = x^B + q^B\, s + h_{0,0}(t) + 2\, H_{0,1}(t)\, \tilde{\stateB}^
\top + \tilde{\stateB}\, H_2(t)\, \tilde{\stateB}^\top
\end{equation}
where $\tilde{\stateB} = \begin{pmatrix} q^B & q^I & \gamma & \hat{\alpha} & \nuU \end{pmatrix}$, and $h_{0,0}:[0,T]\to \Real$, $H_{0,1}:[0,T]\to \mat{1\times 5}$, and $H_2:[0,T]\to \mat{5\times 5}$  are deterministic functions of time reported in \ref{calculations for hjb with alternative filter}. \\

The existence of a classical solution to the HJB equation for this problem depends on a $3\times 3$ matrix Riccati equation reported in \eqref{eq: matrix riccati hjb alternative filter}.

\begin{theorem}
    If there is a solution to the matrix Riccati equation \eqref{eq: matrix riccati hjb alternative filter}, the HJB equation \eqref{eq: hjb equation alternative filter} admits a unique solution in $C^{1,2}([0,T], \Real^{|\stateB|})$ with quadratic growth in $\stateB$. Moreover, the broker's Markovian optimal trading rate for the control problem \eqref{eq: broker control problem} in feedback form is given by
    \begin{equation}
        \nuBstar_t = \hf_0(t) + (\permB + \hf_1(t))\, \QB_t + \hf_2(t)\, \QI_t + \hf_3(t)\, \gamma_t + \hf_4(t)\, \hat \alpha^{\text{alt}}_t + \hf_5(t)\, \nuU_t\,,
    \end{equation}
    where
    \begin{equation}
        \hf_{i}(t) = \frac{h_{1,i}(t)+\GF_1(t)\, h_{3, i}(t)}{\tempB}, \quad i=0, \dots, 5
    \end{equation}
    and the functions $h_{i,j}$ are given in \ref{calculations for hjb with alternative filter}.
\end{theorem}

\section{Numerical results}\label{sec:numerical_results}

We employ the following values of model parameters that largely follow \cite{cartea2022broker}. For the price process, the value of the permanent price impact parameter is $\permB = 10^{-3}$, the volatility is $\sigma^S = 1$, and the initial value of the price is $S_0 = 100$. The parameter of the signal are as follows: the mean-reversion speed is $\kappa^\alpha = 5$, the noise of the signal $\sigma^\alpha = 1$, and the initial value is $\alpha_0 = 0$. The correlation between the noise in the price and the signal is zero  ($\rho=0$). For the uninformed trading flow, we set $\kappa^U = 15$, $\sigma^U = 100$, and $\nuU_0=0$. The temporary price impacts for the broker, informed, and uninformed trader are  $\tempB = 2.1\times 10^{-3}$, $\tempI = 2\times 10^{-3}$, and $\tempU = 2\times 10^{-3}$, respectively.

In the informed trader's learning equations, we set $\sigma^B = 60$ and $\theta^B = 10$. We assume that the initial estimate $\hat{\nuB}_0$ is zero and we set the initial conditional variance $\varI_0 = 0$. In the broker's learning equations, we use the values of the true parameters that describe the dynamics of $\alpha$, and she sets $\alpha_0=0$ and $\varB_0 = 0$. Both  agents employ similar risk aversion parameters; we set  $\betaI = \betaB = 10^{-1}$, $\betaaI = \betaaB = 10^{-3}$, $\rhoI = \rhoB = 10^{-3}$, and $\rhooI = \rhooB = 10^{-5}$. In \ref{sec: Nash}, we discuss how an equilibrium of beliefs could be reached between the two players.
We run numerical simulations where we partition the time horizon $[0,T]$, with $T=1$, into $1,000$ equally-spaced intervals and we run $10,000$ simulations.\footnote{The code is available at \url{https://github.com/muhammadalifaqsha/broker_informed_noise_filtering_game}.}
Our parameters satisfy the conditions   of Proposition \ref{prop: freiling riccati existence} for existence and uniqueness of a solution. For more details see \ref{app: details existence uniqueness eigenvalues}.

Figure \ref{fig:04_01} shows, for each learning strategy, one sample path for (i) the midprice $S_t$, (ii) the signal $\alpha_t$, (iii) the trading speeds of the broker, the informed, and the uninformed trader, (iv) the cash processes of the informed trader and the broker, and  (vi) the inventory of both the informed trader and the broker.
For each process, we  plot  the $90\%$ confidence bands across simulations. As expected, for both strategies, the trajectory of the inventory processes finish near zero so agents avoid paying the terminal inventory penalties. On the other hand, when the broker employs the alternative learning strategies, her trading speeds resemble those of the informed trader. This is due to the role of the signal that plays into both players' strategies; indeed, in the first strategy, the broker's estimates of the signal are quite poor. We also find that if one were to use the first strategy while employing the alternative estimates of $\alpha$, the resulting strategy will resemble the second strategy.

\begin{figure}[h]
    \centering
    \includegraphics[width=0.85\textwidth]{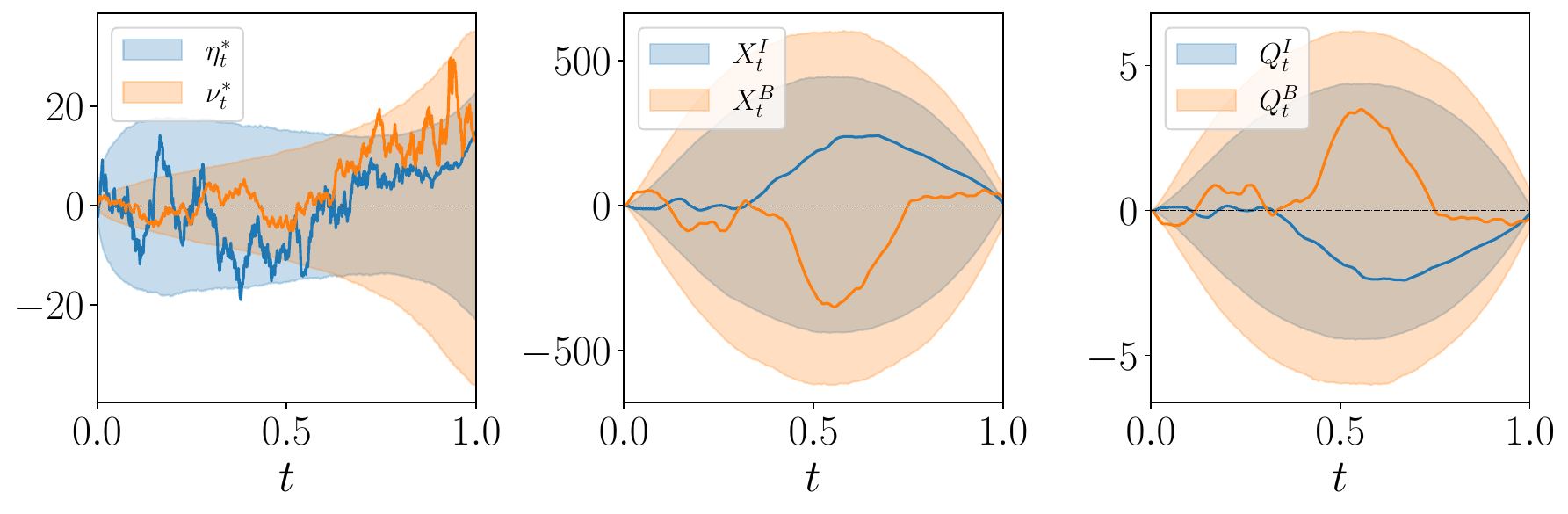}\\
    \includegraphics[width=0.85\textwidth]{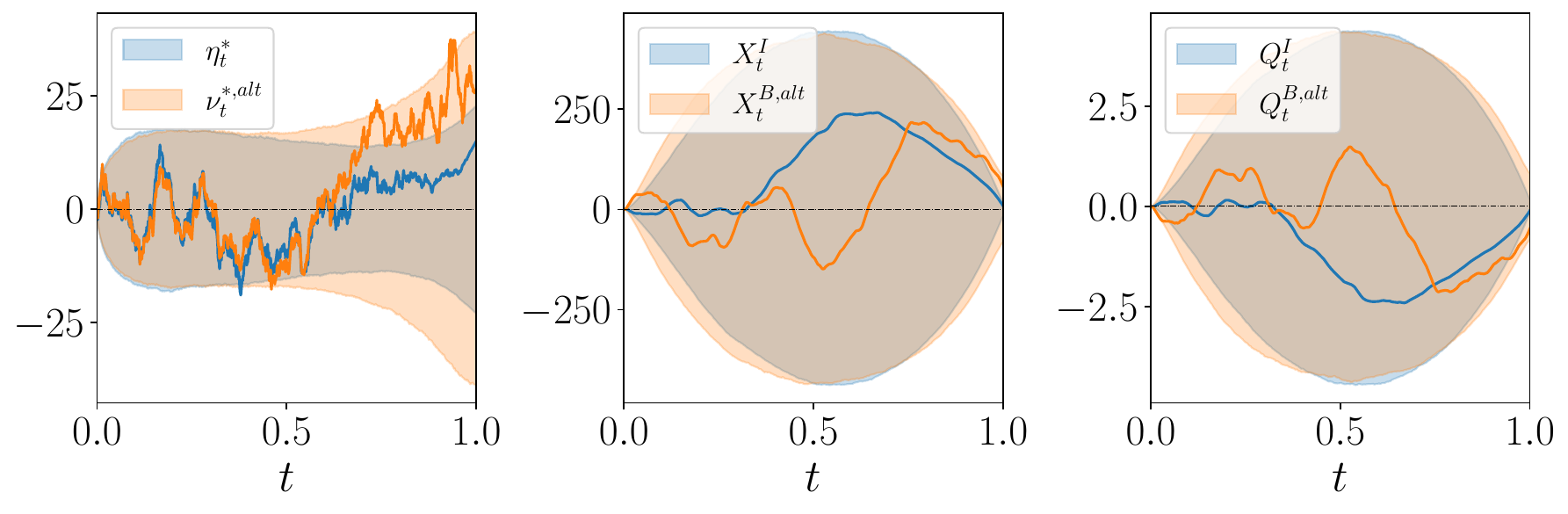}
    \caption{One sample path and $90\%$ confidence bands across simulations for the following processes. The broker's trading speed $\nuB_t$, the trading speed $\nuI_t$ of the informed trader, the cash process $\XI$ of the informed trader, the cash process  $\XB$ of the broker, the inventory $\QI$ of the informed trader, and the inventory $\QB$ of the broker. The top panel is when the broker filters from prices and the bottom panel is when she filters from the informed trader's flow.}
    \label{fig:04_01}
\end{figure}

Next, Figure \ref{fig:04_03} shows the estimates $\hat{\alpha}$ and $\hat{\nuB}$ constructed by the broker and the informed trader, respectively. The top panels are for when the broker filters prices and the bottom panels are for when she filters the informed trader's flows. We plot the  $90\%$ confidence bands across simulations. The bands around the filter are not the conditional variance of the filter, instead, they are the $90\%$ confidence bands for the filtered processes across simulations.

\begin{figure}[H]
    \centering
    \includegraphics[width=0.6\textwidth]{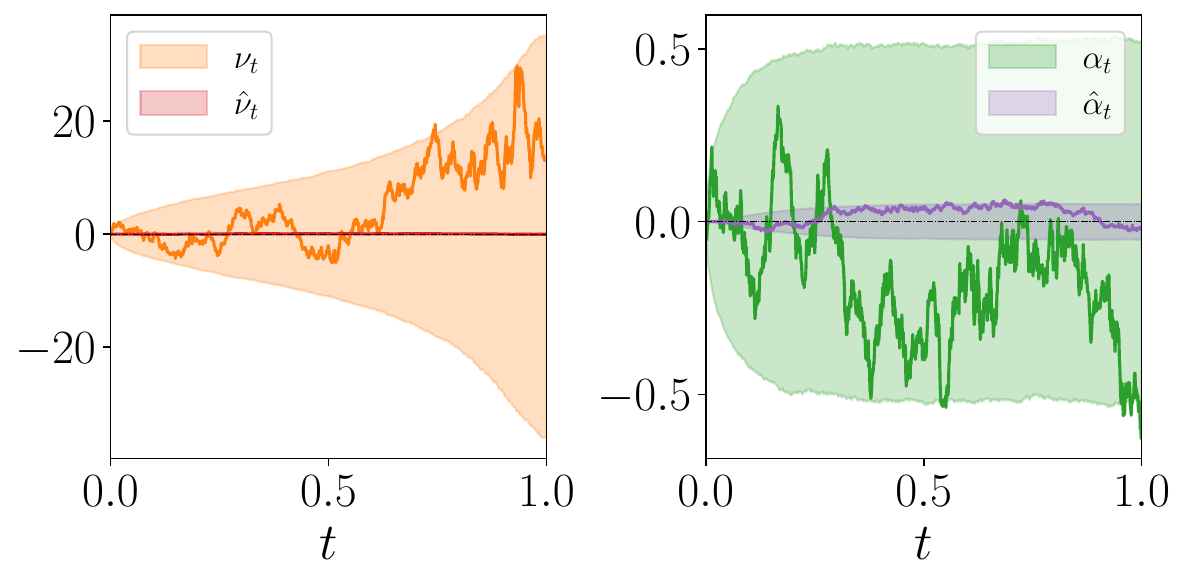}\\
    \includegraphics[width=0.6\textwidth]{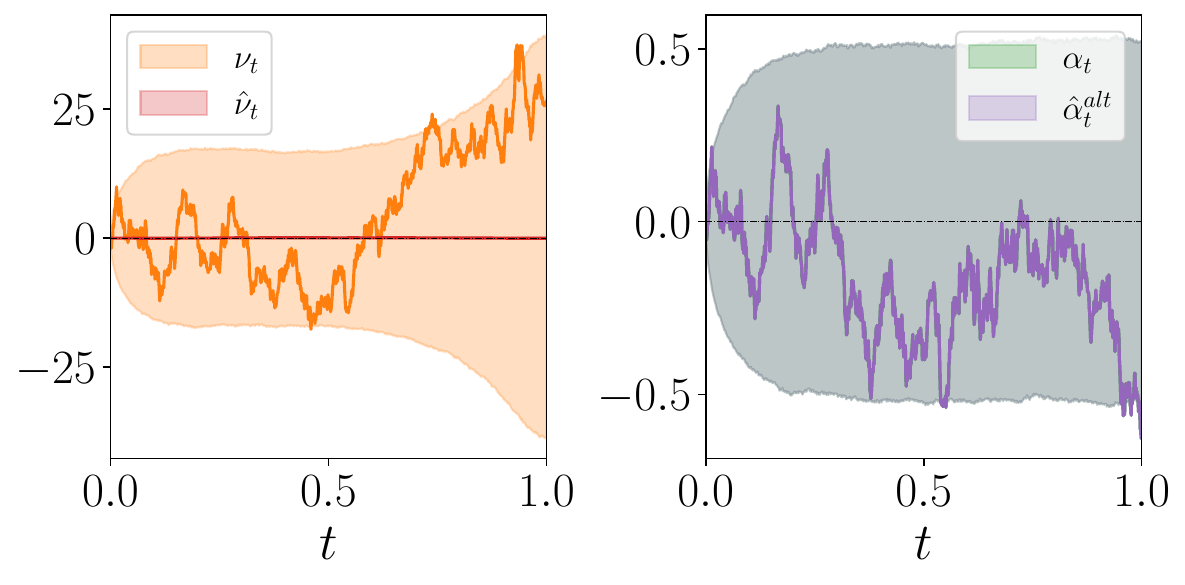}
    \caption{One realisation of filtering carried by both informed trader and broker. The top panel is when the broker filters from prices and the bottom panel is when she filters from the informed trader's flow.}
    \label{fig:04_03}
\end{figure}

For both cases, the informed trader has poor estimates of the broker's activity in the lit market. This is ultimately due to the level of the permanent impact  $\permB$ of the broker's trading activity. On the other hand, the broker has poor estimates of the signal when she filters from prices, but she predicts the signal almost perfectly when she filters from the informed trader's flow. The accuracy of our estimators are in line with practice; in financial markets, both agents face a low signal-to-noise ratio when using prices as a source of information to filter their competitor's private information.

Lastly, Figure \ref{fig:04_09} shows the components of the trading strategy of the broker. The main drivers are (i) inventory management (left panel), (ii) unwinding the uninformed order flow (third panel), and (iii) adjusting to the  inventory level of the informed trader (fourth panel). When the broker is filtering from the informed trader's strategy, she includes an additional component $\gamma$, which is the informed trader's flows after removing his inventory risk aversion component. As shown in the figure, the effect of $\gamma$ is negligible compared to the other four components. This shows that after correctly estimating the signal, the broker does not externalise the informed trader's flow.

\begin{figure}[H]
    \centering
    \includegraphics[width=0.9\textwidth]{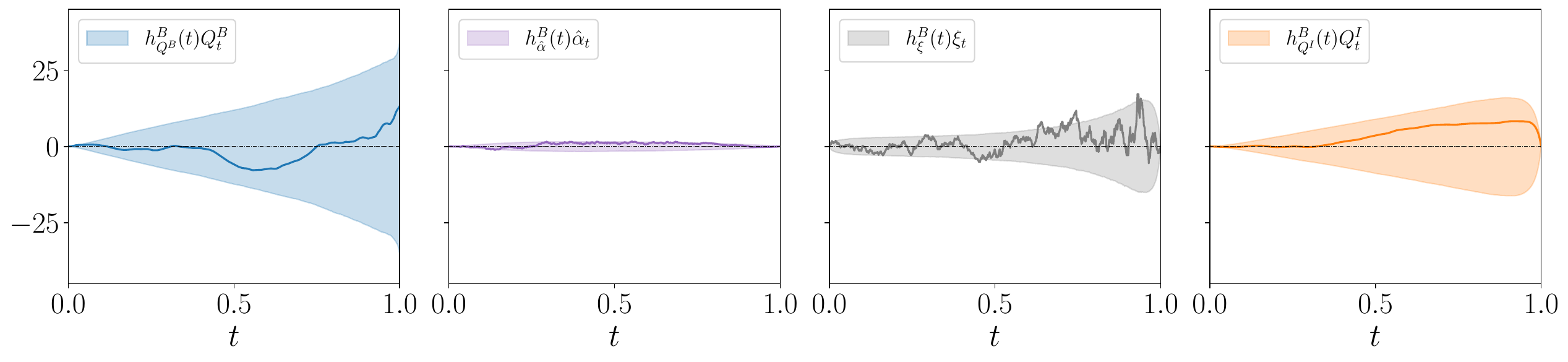}\\
    \includegraphics[width=0.9\textwidth]{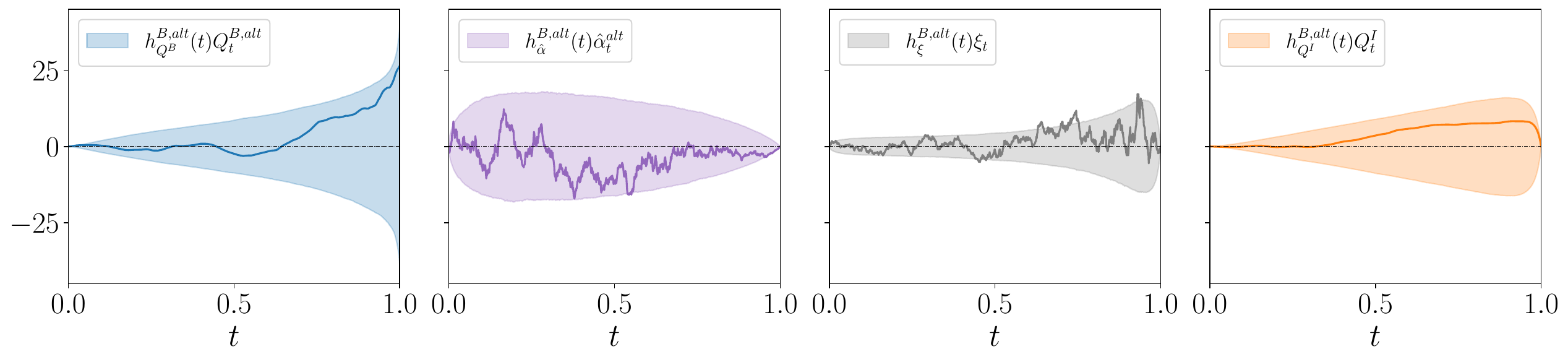}\\
    \includegraphics[width=0.25\textwidth]{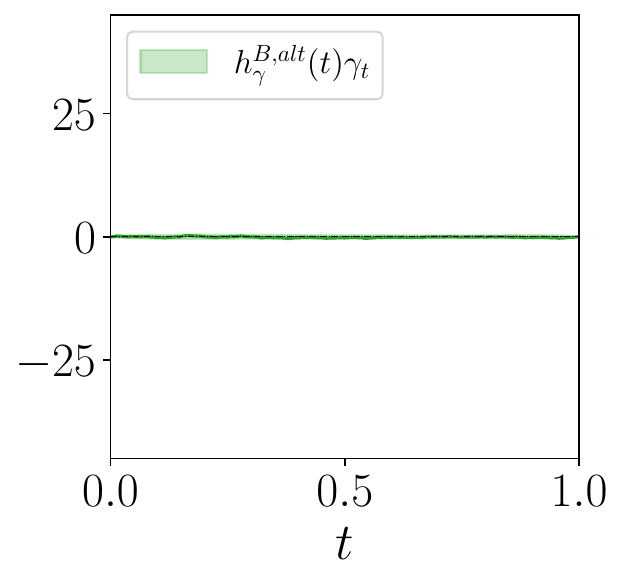}
    \caption{The additive components of the broker's strategy. The function $h_{(\cdot)}^B$ refers to the coefficient of component $(\cdot)$ in the strategy; the superscript ``$\alt$'' refers to the second strategy with the alternative filter. The first row is showing the four additive components when the broker filters from prices. The second and third row are showing the five additive components when the broker filters from the informed trader's flow.}
    \label{fig:04_09}
\end{figure}

Moreover, the last column of Figure \ref{fig:04_09} and the last panel of Figure \ref{fig:04_01} show that the broker trades against the inventory level of the informed trader, in particular, she externalises a significant portion of the informed trading flow systematically.
The second column of Figure \ref{fig:04_09} shows the speculative component of the broker's strategy which employs the estimated signal $\hat{\alpha}$. Observe that $\hat{\alpha}$ does not contribute to the broker's strategy to the same extent as the other three components due to the poor estimation accuracy of the filter. On the other hand, the speculative component from exploiting $\hat{\alpha}^{\text{alt}}$ is more predominant than in the first strategy.

\subsection{Performance and comparison against  benchmarks}

Below, we use $Q_t^{B,(i)}$, for $i\in\{1,2,3\}$, to denote  the inventory of the broker when she follows the benchmark strategy $\nuB_t^{(i)}$.
We employ the following  benchmarks.

\begin{enumerate}[label=(\roman*)]
    \item Benchmark 1: externalise the trading flow of  the informed trader, internalise the trading flow of the uninformed trader, and unwind remaining inventory with TWAP. For this benchmark, the trading speed of the broker is:
    \begin{equation}
        \nuB_t^{(1)} = \nuI_t^{*} - \frac{Q_t^{B,(1)}}{T-t}\,.
    \end{equation}
    \item Benchmark 2: internalise the trading flow of both the informed and uninformed traders, and unwind the broker's inventory with TWAP. For this benchmark, the trading speed of the broker is:
    \begin{equation}
        \nuB_t^{(2)} = - \frac{Q_t^{B,(2)}}{T-t}\,.
    \end{equation}
    \item Benchmark 3: externalise the trading flow of both the informed and uninformed traders. For this benchmark, the trading speed of the broker is:
    \begin{equation}
        \nuB_t^{(3)} = \nuI_t^{*} + \nuU_t\,.
    \end{equation}
    
\end{enumerate}

We compute outperformance,  defined as the additional trading performance from following the strategy  \eqref{eq:01_informed_strategy_feedback_form}, when 
compared to trading according to the benchmark strategies described above. Outperformance is expressed in dollars per million dollars traded. Mathematically, the outperformance compared to the $i$-th benchmark is  expressed as
\begin{equation}
    \widehat{\text{Out}}^{(i)} = \frac{1}{N}\sum_{n=1}^N \left[ \frac{\XBstar_{T,n} + \QBstar_{T,n}\, S_{T,n}^{*} - \left(\XB[B,(i)]_{T,n} + \QB[B, (i)]_{T,n}\, S_{T,n}^{(i)}  \right)}{\int_0^T S_{u,n}^{(i)}\, \left(|\nuB_{u,n}^{(i)}| + |\nuI_{u,n}^{*}| + |\nuU_{u,n}|\right)\, \d u} \right]\times 10^6\,,
\end{equation}
where $N=10,000$ simulations,  $i\in\{1,2,3\}$ refers to the benchmark strategy, and the second sub-index is  the simulation number.
Table \ref{tab:04_01} shows the outperformance for each benchmark together with a $p$-value
for the one-sided $t$-test with $H_0:\, \text{Out}^{(i)}= 0$ against $H_1:\, \text{Out}^{(i)}> 0$.

\begin{table}[H]
\centering
\begin{tabular}{r|r|r}
\noalign{\vskip 1mm}
\hline\hline
\noalign{\vskip 1mm}
$i$ & $\widehat{\text{Out}}^{(i)}$ (std) & $p$-value             \\ 
\noalign{\vskip 1mm}
\hline\hline
\noalign{\vskip 1mm}
1 & 0 (570)             & $> 0.5$               \\
2 & 38 (354)             & $<0.001$ \\ 
3 & 95 (696)             & $<0.001$  \\
\noalign{\vskip 1mm}
\hline\hline
\noalign{\vskip 1mm}
\end{tabular}
\hspace{5em}
\begin{tabular}{r|r|r}
\noalign{\vskip 1mm}
\hline\hline
\noalign{\vskip 1mm}
$i$ & $\widehat{\text{Out}}^{(i)}$ (std) & $p$-value             
\\ 
\noalign{\vskip 1mm}
\hline\hline
\noalign{\vskip 1mm}
1 & 36 (117)             & $<0.001$               \\
2 & 78 (389)             & $<0.001$ \\ 
3 & 127 (518)             & $<0.001$  \\
\noalign{\vskip 1mm}
\hline\hline
\noalign{\vskip 1mm}
\end{tabular}
\caption{Model outperformance with respect to the three benchmarks and their $p$-values from a one sided $t$-test ($H_0:\,\text{Out}^{(i)} = 0$ and $H_1:\, \text{Out}^{(i)} > 0$). The left table is the first strategy when the broker filters from prices; the right table is when she employs the alternative filter.
The raw performance of the optimal strategy employing the first learning method is 0.32 (2.81). For the alternative learning equation, the performance is 0.45 (2.08). The performance for three benchmarks are 0.32 (2.39), 0.19 (2.71), and -0.04 (0.21).
}
\label{tab:04_01}
\end{table}

Due to the low signal-to-noise ratio, the performance of the first strategy is indistinguishable from that of the first benchmark strategy which externalises the informed flow, internalises the uninformed flow, and unwinds inventory linearly. Note that we do not reject the null hypothesis (that the means are the same) in the $t$-test.

On the other hand, the second strategy outperforms all the benchmarks significantly due to better estimates of the signal. We remark that the estimate  \eqref{eq: alpha mean estimate} is only valid if (i) the informed trader carries out his trading exclusively with the broker and (ii) that he does not hide information or mislead the broker. If this is not the case, then the inventory of the informed trader, from the perspective of the broker, is mispecified, and the calculation of $\hat{\alpha}^{\text{alt}}_t$ above blows up at the terminal time $T$ as the limit in \eqref{eq: limit of Z tilde} blows up due to unaccounted inventory differences.

Figure \ref{fig:04_08a} shows a sample path when $\QI[I]_0 \sim \mathcal{N}(0, 1)$ and the broker assumes that $\QI[I]_0=0$. The mispecification of inventory has less impact on the accuracy of the estimate $\hat{\alpha}^{\text{alt}}$ in the beginning of the trading window because the informed trader trades less aggressively at this stage. However, as trading progresses, the accuracy deteriorates. 
Therefore, if the broker employs $\hat{\alpha}^{\text{alt}}$, then she must rely less on the signal less as the trading nears  its close, since inventory mispecification increasingly contaminates the extracted signal the later stages.

\begin{figure}[H]
    \centering
    \includegraphics[width=0.4\textwidth]{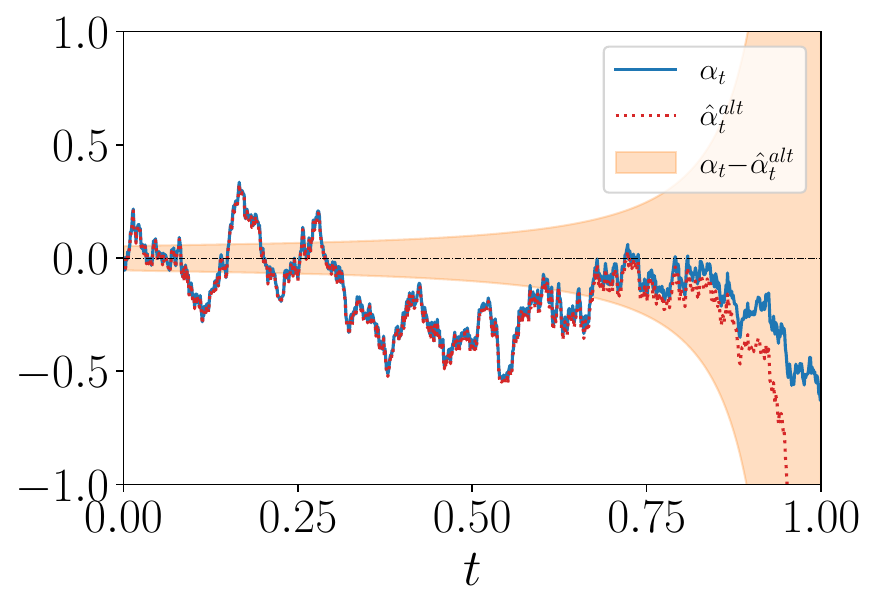}
    \caption{One sample path for the signal $\alpha$ and the new estimates $\hat{\alpha}$ when the broker mispecifies the informed trader's inventory.}
    \label{fig:04_08a}
\end{figure}

\subsubsection{Naive estimate for the signal assuming near-zero values for $\hat \nuB$}

In practice, given that the permanent impact parameter $\permB$ of the broker's trading activity in the market is small, the noise driving price movements obscures the filtered order flow of the broker. Thus, the contributions of the filtered trading speed $\hat{\nuB}$ in the informed trader's strategy are small enough so that they may be assumed to be zero. We illustrate this  in  Figure \ref{fig:04_07} (see middle panel).

\begin{figure}[H]
    \centering
    \includegraphics[width=0.8\textwidth]{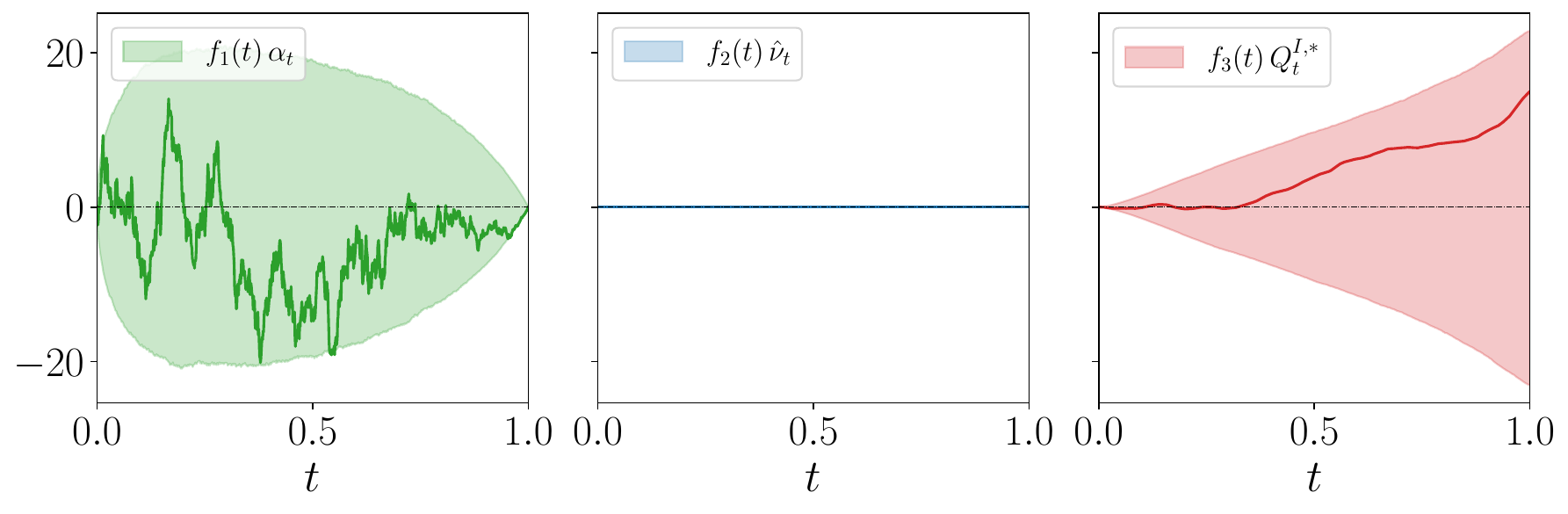}
    \caption{The components of the informed trader's strategy.}
    \label{fig:04_07}
\end{figure}

The broker could obtain a better, near-perfect, estimate for $\alpha$ given that  the contribution of $\hat{\nuB}$ in the informed trader's strategy is near-zero. More precisely, the broker could ignore those contributions (by assuming $f_2(t)\, \hat{\nuB}_t = 0$) and obtain the following alternative estimates for $\alpha$
\begin{equation}
    \hat{\alpha}^{\textup{naive}}_t := \underbrace{ \alpha_t + \frac{f_2(t)}{f_1(t)}\, \hat{\nuB}_t }_{ \approx \alpha_t }= \frac{\nuI_t^{*} - f_3(t)\, Q_t^{I,*}}{f_1(t)}\,.
\end{equation}
The calculation for $\hat{\alpha}^{\textup{naive}}_t$ above always converges as $t\to T$ because $$\lim_{t\to T} \frac{f_2(t)}{f_1(t)} = \lim_{t\to T} \frac{f_2^{'}(t)}{f_1^{'}(t)} = \permB$$ using \eqref{eq:trader_z1}--\eqref{eq:trader_z2}.  Similar to the filter $\hat{\alpha}^{\text{alt}}$, we remark that this estimate is only valid if the informed trader carries out his trading exclusively with the broker. If this is not the case, then the inventory of the informed trader, from the perspective of the broker, is mispecified, and the calculation of $\hat{\alpha}_t$ above blows up at the terminal time $T$.

\subsection{Second-order effects}
In the problem formulation above, both agents act strategically, though the informed trader’s strategy is only strategic to a first order. Specifically, the informed trader does not account for the possibility that the broker may filter his private signal. Conversely, the broker operates under the assumption that the informed trader will incorporate her trading activity to enhance trading performance. In \ref{app: second order} we investigate how to adapt our framework to accommodate these second-order effects.

\subsection{Mispecification of learning parameters}

Next, we examine how mispecification of the value of model parameters affect the outperformance metrics reported above. Using the values reported in Table \ref{tab:04_01} as a baseline, we stress each parameter by increasing and decreasing its value by $50\%$, and report the results in Table \ref{tab:mispecification_broker} when the broker filters from prices and Table \ref{tab:mispecification_broker_alternative} when she filters from the informed trader's flows.

\begin{table}[H]
\centering
\begin{tabular}{r|r|r|r|r|r}
\noalign{\vskip 1mm}
\hline\hline
\noalign{\vskip 1mm}
$i$ & Base & $\kappa^\alpha$ (+50\%) & $\kappa^\alpha$ (-50\%) & $\sigma^\alpha$ (+50\%) & $\sigma^\alpha$ (-50\%)    \\ 
\noalign{\vskip 1mm}
\hline\hline
\noalign{\vskip 1mm}
1 & 0 (570)             & 0 (571) & -1 (570) & -1 (570) & 0 (573)             \\
2 & 38 (354)*              & 38 (354)*  & 37 (362)*  & 38 (362)*  & 38 (355)* \\ 
3 & 95 (696)*              & 95 (697)*  & 95 (695)*  & 95 (695)*  & 95 (697)* \\
\noalign{\vskip 1mm}
\hline\hline
\noalign{\vskip 1mm}
\end{tabular}\\
\begin{tabular}{r|r|r|r|r|r}
\noalign{\vskip 1mm}
\hline\hline
\noalign{\vskip 1mm}
$i$ & Base & $\theta^B$ (+50\%) &  $\theta^B$ (-50\%) & $\sigma^B$ (+50\%) & $\sigma^B$ (-50\%)     \\ 
\noalign{\vskip 1mm}
\hline\hline
\noalign{\vskip 1mm}
1 & 0 (570)             & 0 (569) & 0 (573) & -1 (573) & 0 (566)             \\
2 & 38 (354)*              & 39 (352)*  & 38 (358)*  & 38 (358)*  & 39 (348)* \\ 
3 & 95 (696)*              & 95 (695)*  & 95 (697)*  & 95 (697)*  & 96 (693)* \\
\noalign{\vskip 1mm}
\hline\hline
\noalign{\vskip 1mm}
\end{tabular}
\caption{Outperformance of model to the three benchmarks when the broker mispecified her learning parameters for the first strategy. Significant cases are marked with *.}
\label{tab:mispecification_broker}
\end{table}

\begin{table}[H]
\centering
\begin{tabular}{r|r|r|r|r|r}
\noalign{\vskip 1mm}
\hline\hline
\noalign{\vskip 1mm}
$i$ & Base & $\kappa^\alpha$ (+50\%) & $\kappa^\alpha$ (-50\%) & $\sigma^\alpha$ (+50\%) & $\sigma^\alpha$ (-50\%)    \\ 
\noalign{\vskip 1mm}
\hline\hline
\noalign{\vskip 1mm}
1 & 36 (117)*             & 33 (153)* & 36 (181)* & 35 (200)* & 15 (209)*            \\
2 & 78 (388)*             & 74 (512)*  & 78 (315)*  & 77 (282)*  & 53 (627)*  \\ 
3 & 127 (518)*              & 124 (529)*  & 126 (528)*  & 126 (532)*  & 108 (549)* \\
\noalign{\vskip 1mm}
\hline\hline
\noalign{\vskip 1mm}
\end{tabular}\\
\begin{tabular}{r|r|r|r|r|r}
\noalign{\vskip 1mm}
\hline\hline
\noalign{\vskip 1mm}
$i$ & Base & $\theta^B$ (+50\%) &  $\theta^B$ (-50\%) & $\sigma^B$ (+50\%) & $\sigma^B$ (-50\%)     \\ 
\noalign{\vskip 1mm}
\hline\hline
\noalign{\vskip 1mm}
1 & 36 (117)*             & 36 (119)* & 37 (113)* & 37 (112)* & 36 (122)*             \\
2 & 78 (388)*              & 78 (383)*  & 78 (404)*  & 78 (408)*  & 77 (376)* \\ 
3 & 127 (518)*              & 126 (518)*  & 127 (518)*  & 127 (518)*  & 126 (518)* \\
\noalign{\vskip 1mm}
\hline\hline
\noalign{\vskip 1mm}
\end{tabular}
\caption{Outperformance of model to the three benchmarks when the broker mispecified her learning parameters for the alternative strategy. Significant cases are marked with *.}
\label{tab:mispecification_broker_alternative}
\end{table}

We observe that, even when the broker mispecifies the parameters when filtering from prices, her outperformance levels remain approximately the same as in the baseline. Thus, benchmark 1 remains a strong candidate in this scenario. However, when the broker filters from the informed trader's flows, the optimal strategy outperforms all the benchmark, even with parameter mispecification. Thus, filtering from informed trader's flows remains a strong strategy for the broker.

\subsection{The role of externalisation and stressing model parameters}

In this section, we examine the role of externalisation in the trading strategy of the broker. Additionally, we analyse the impact of changes in model parameters on the externalisation rate. Specifically, we focus on (i) the parameters in the dynamics of the signal and (ii) the parameters of the informed trader's assumption over the broker's trading strategy.

Figure \ref{fig: stressing  externalisation} shows the ``effective'' externalisation quotient for stressed values of model parameters.
To derive this quotient, observe that $\nuBstar_t = f^B_\alpha(t)\,\hat\alpha_t + \epsilon_t$ and similarly, that $\nuIstar_t = f^I_\alpha(t) \,\alpha_t +  \varepsilon_t$ for some $\{\epsilon_t,\varepsilon_t\}$. Thus, a straightforward calculation shows that if $\alpha_t \approx \hat\alpha_t$, then $\nuBstar_t = f^B_\alpha(t) / f^I_\alpha(t) \,\nuIstar_t + \upsilon_t$. In this setup,  the quotient $f^B_\alpha(t) /f^I_\alpha(t)$ represents the effective externalisation rate at which the broker offloads the trading of the informed trader. This quotient is representative of the  externalisation rate in practice when the broker has a good estimate of the private signal.

\begin{figure}[!h]
    \centering
    \includegraphics[width=0.45\textwidth]{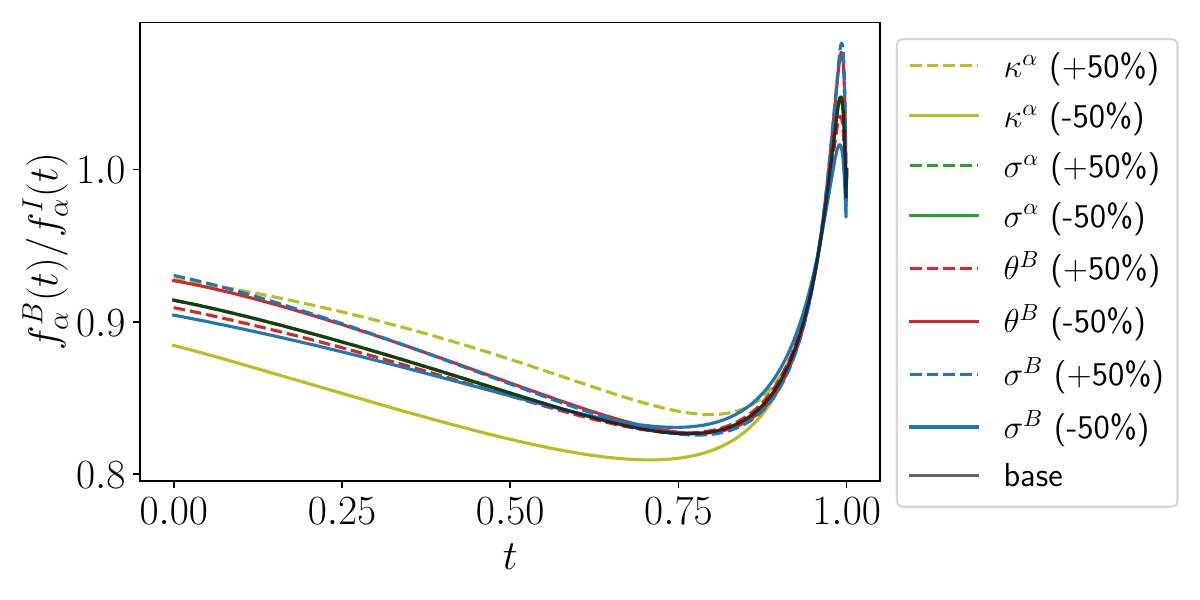}
    \caption{Quotient $f^B_\alpha /f^I_\alpha$ that represents the effective externalisation followed by the broker. This quotient accounts for the speculation by the broker according to her signal estimates.}
    \label{fig: stressing  externalisation}
\end{figure}

Figure \ref{fig: stressing  externalisation}  shows that the mean-reversion parameter $\kappa^\alpha$ in the signal dynamics influences the effective externalisation rate significantly. Large values of $\kappa^\alpha$ result in a signal that diffuses less. As such, the broker trusts her estimates more and increases her speculation part, thus she externalises more. As the signal diffuses less (from zero), its role  in the strategies of both the broker and informed trader diminishes; see Figure \ref{fig: stressing alpha coefficients}. Visually, the signal appears to impact both the broker and informed trader's strategies similarly. However, the broker's  sensitivity to the signal decreases at a faster rate when $\kappa^\alpha$ is low (i.e. when the signal diffuses more). This can be attributed to the broker’s growing uncertainty in her estimate of the signal due to accumulated errors. When $\kappa^\alpha$ is high (i.e. when the signal diffuses less), the broker's  sensitivity to the signal decreases at a more equal rate to the trader. Consequently, the broker increases her externalisation rate without concern for the risks associated with signal misspecification.

\begin{figure}[!h]
    \centering
    \includegraphics[width=0.7\textwidth]{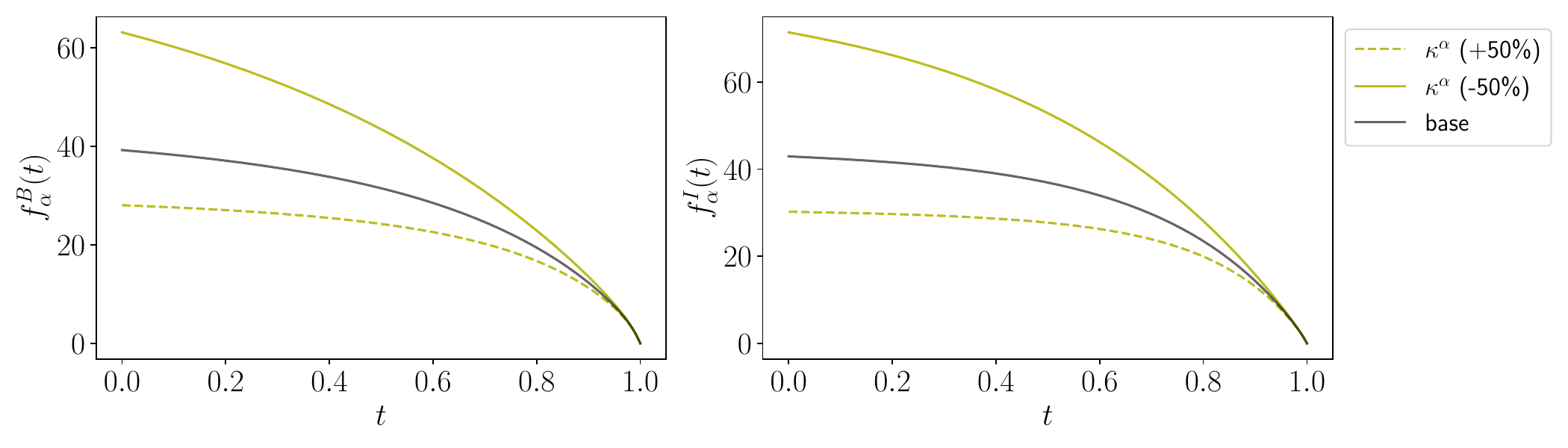}
    \caption{The coefficients entering the estimated signal component in the broker (left panel) and informed trader's strategy (right panel).}
    \label{fig: stressing alpha coefficients}
\end{figure}

Furthermore, when the value of $\sigma^B$ in the informed trader's model \eqref{eq: nuB_according_to_trader} is large (or when $\theta^B$ is small), the informed trader believes that the broker will often trade aggressively. In response,  the broker increases her externalisation rate.
To assess the proportion of the average informed flow that the broker externalises, and following the approach in \cite{bergault2024mean}, we compute the proxy $G(\nuBstar_t) / G(\nuIstar_t)$ of the quotient $\nuBstar_t / \nuIstar_t$. The function $G$ is defined as  $G(x) = \max(x,\epsilon)\,\mathds{1}(x\geq 0) + \min(x, -\epsilon)\,\mathds{1}(x< 0)$, with  $\epsilon = 0.1$. The function $G$ prevents the ratio $\nuBstar_t / \nuIstar_t$ from becoming  undefined or excessively large when $\nuIstar_t$ is close to zero.

\begin{figure}[!h]
    \centering
    \includegraphics[width=0.45\textwidth]{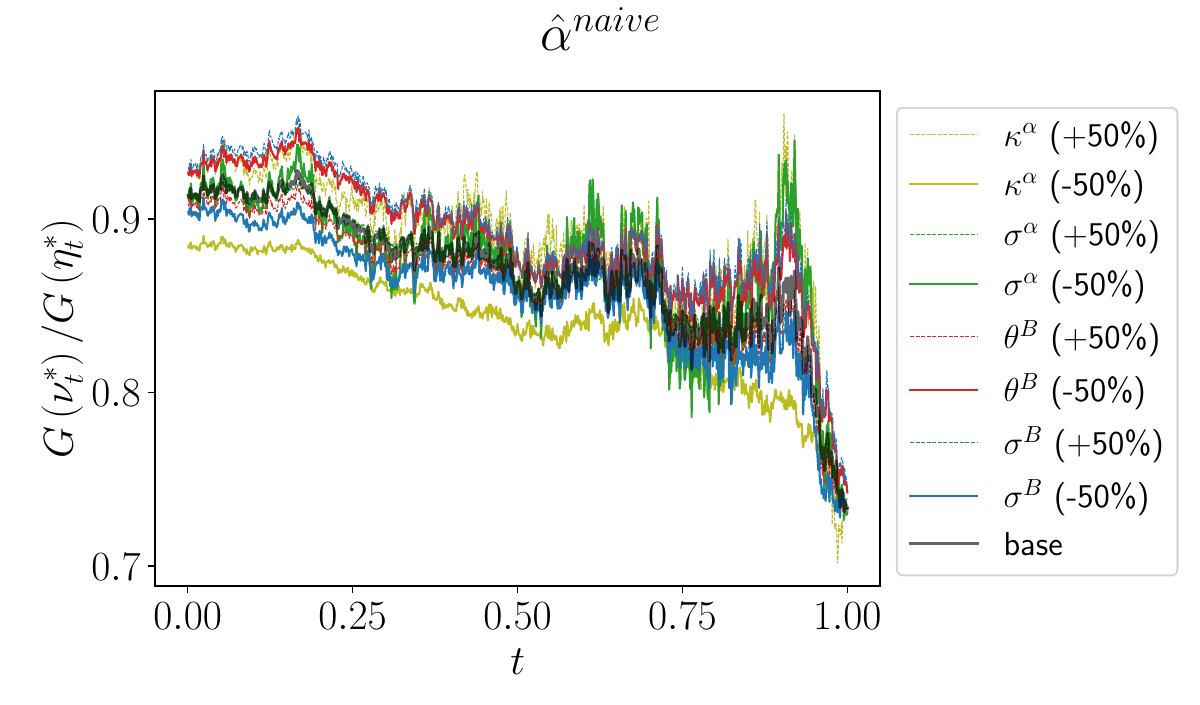}
    \caption{The left panel is for when the broker uses the original estimates $\hat\alpha$, and the right panel is for the naive filter $\hat\alpha^{\text{naive}}$. In both panels we show the median trajectory when stressing the values of model parameters.
    }
    \label{fig: stressing filtering externalisation}
\end{figure}

Figure \ref{fig: stressing filtering externalisation} illustrates that, for the naive case, the externalisation is ranging between $70\%$ and $90\%$. The ordering of the lines when stressing $\kappa^\alpha$ in the right panel of Figure \ref{fig: stressing filtering externalisation} is similar to those in Figure \ref{fig: stressing  externalisation}.

\subsection{Discussion: equilibrium of beliefs}\label{sec: Nash}
In practice, both the broker and the informed trader will adjust their beliefs with time. The broker believes that the signal is characterised by parameters $\kappa^{\alpha,B}$ and $\sigma^{\alpha,B}$ (not necessarily the true parameters $\kappa^\alpha$ and $\sigma^\alpha)$, similarly, the informed trader models the trading speed of the broker as an OU process with parameters $\theta^B$ and $\sigma^B$. 
Similar to Section 4 in \cite{donnelly2025liquidity}, 
one could then look at the equilibrium of beliefs. More precisely, if we $\mathfrak{b}^B = (\kappa^{\alpha,B},\sigma^{\alpha,B})$ and $\mathfrak{b}^I = (\theta^{B},\sigma^{B})$ be the beliefs of the informed trader and the broker respectively, we  could study the performances of each player as they tweak these beliefs in search of an equilibrium point. We envisage that such an exercise would yield equilibrium beliefs $\mathfrak{b}^{B*}$  and $\mathfrak{b}^{I*}$ with the property that (i) the broker's beliefs are close to the true parameters of the signal and (ii) the informed trader's beliefs fit well the sample paths of the equilibrium trading speed of the broker.
This investigation is an interesting avenue for future research.

\section{Conclusions}
In this paper we computed the optimal strategies of 
a broker and an informed trader when  both parties filter the information they do not have. Brokers hold a slight advantage 
over informed traders because informed traders rely
solely on prices to filter information, whereas 
brokers have prices and order flow at their disposal.
We find that when brokers filter the stochastic 
drift $\alpha_t$ from prices, the performance of the optimal strategy 
is indistinguishable from a naive strategy  that externalises the order flow from the informed, internalises the uninformed flow, and unwinds inventory according to a TWAP schedule. On the other hand, if brokers use the order flow to filter the stochastic drift, the outperformance over the benchmarks is in the order of magnitude of transaction costs.

\section*{Funding:}
\noindent AA's research is supported by  the Oxford-Man Institute of Quantitative Finance through the EPSRC Centre for Doctoral Training in Mathematics of Random Systems: Analysis, Modelling and Simulation (ESPRC Grant EP/S023925/1).\\

\section*{Disclosure statement}

The authors have no conflicts of interest to declare that are relevant to the content of this article.

\noindent For the purpose of open access, the authors have applied a CC BY public copyright license to any author accepted manuscript arising from this submission.

\vspace{0.5cm}
\bibliographystyle{apalike}
\bibliography{references}

\newpage
\appendix

\section{Derivations for the control problem of the informed trader}\label{app: derivations informed}
Here we collect the calculations associated with the control problem of the informed trader.

\subsection{Hamilton–Jacobi–Bellman equation} 
The Hamilton–Jacobi–Bellman (HJB) equation that characterises the value function $H^I$ is 
\begin{align}
  0 &= \partial_t H^I - (\rho_0 + \rho_1\, \varI_t) \left(q^I\right)^2  +    \sup_{\nuI}  \big[ (\permB\,\hat{\nuB}  + \alpha)\, \partial_s H^I - \nuI 
  \left( s + \tempI\, \nuI  \right)\, \partial_{x^I} H^I + \nuI\, \partial_{q^I} H^I \big] \nonumber \\
  & \qquad - \kappa^\alpha \alpha\, \partial_\alpha H^I - \theta^B\,\hat{\nuB})\, \partial_{\hat{\nuB}} H^I + \frac{1}{2}(\sigma^S)^2\, \partial_{ss} H^I \nonumber\\
  & \qquad + \permB\, \varI_t\, \partial_{s\hat{\nuB}}H^I +
  \frac{1}{2} (\sigma^{\alpha})^2 \partial_{\alpha\alpha}H^I + \frac{1}{2} \left(\frac{\permB\, \varI_t}{\sigma^S}\right)^2
  \partial_{\hat{\nuB} \hat{\nuB}}H^I \\
  &\qquad +\sigma^S \,\sigma^{\alpha}\, \rho\, \partial_{s\alpha} H^I + \frac{\permB\, \sigma^{\alpha}\, \varI_t\, \rho}{\sigma^S}\, \partial_{\alpha\hat{\nuB}}H^I\,, \label{eq:trader_hjb}
\end{align}
with terminal condition 
$$H^I(T,\stateI)=x^I + q^I\,s - \left(\beta_0+ \beta_1\,\varI_T\right)\left(q^I\right)^2\,.$$ 
The supremum in the Hamiltonian in \eqref{eq:trader_hjb} is attained at
\begin{equation}
\label{eq: informed trader optimal non feedback}
\nuI^* = \frac{\partial_{q^I}H^I - \partial_{x^I}H^I\,s}{2\,\partial_{x^I}H^I\,\tempI} \,.
\end{equation}
After substitution, employ the ansatz $$H^I=x^I + h^I(t,s, q^I,\alpha,\hat{\nuB})\,$$ to obtain that $h^I$ satisfies the partial differential equation (PDE)
\begin{align}
  &0 =  - \left(\rho_0 + \rho_1\, \varI_t\right) \left(q^I\right)^2 + \partial_t h^I  + (\permB
  \hat{\nuB} + \alpha)\, \partial_s h^I + \left(\partial_{q^I}h^I\right)^2 / (4\, \tempI) - \kappa^{\alpha} \,\alpha\, \partial_{\alpha} h^I \nonumber\\
  &\qquad - \theta^B\,  \hat{\nu}\, \partial_{\hat{\nuB}}h^I +
  \frac{1}{2} \left(\sigma^S\right)^2 \partial_{ss}h^I  + \permB\, \varI_t\, \partial_{s
  \hat{\nuB}}h^I + \frac{1}{2} (\sigma^{\alpha})^2 \partial_{\alpha\alpha} h^I\nonumber\\
  &\qquad + \frac{1}{2} \left(\frac{\permB\, \varI_t}{\sigma^S}\right)^2 \partial_{\hat{\nuB}  \hat{\nuB}}h^I + \frac{\permB\, \sigma^{\alpha}\, \varI_t\, \rho}{\sigma^S}\, \partial_{\alpha \hat{\nuB}}h^I \,,
  \nonumber
\end{align}
with terminal condition 
$$h^I(T, s, q^I, \alpha, \hat{\nuB}) =  q^I\, s - (\beta_0 + \beta_1\, \varI_T) \left(q^I\right)^2 \,.$$

Next, use the ansatz $$h^I(t,s, q^I,\alpha,\hat{\nuB}) = q^I\, s + g_0^I  (t, \alpha,
\hat{\nuB}) \, + q^I\, g_1^I  (t, \alpha, \hat{\nuB}) + \left(q^I\right)^2 \, g_2^I (t)$$ to obtain the system of differential equations
\begin{align}
        0 &= - (\rho_0 + \rho_1 \varI_t) + (g_2^I)' + \frac{(g_2^I)^2}{\tempI} \label{eq:trader_sys_1}\\
        0 &= \partial_{t}\, g_1^I + (\permB \hat{\nuB} + \alpha) + \frac{g_1^I\, g_2^I}{2 \tempI} - \kappa^{\alpha}\, \alpha\, \partial_{\alpha}\, g_1^I  - \theta^B\,  \hat{\nuB}\, \partial_{\hat{\nuB}}\, g_1^I + \frac{1}{2}(\sigma^{\alpha})^2 \partial_{\alpha\, \alpha}\, g_1^I \,, \nonumber\\
        &\phantom{{}={}} + \frac{1}{2} \left(\frac{\permB\, \varI_t}{\sigma^S}\right)^2 \partial_{\hat{\nuB}\, \hat{\nuB}}\, g_1^I + \frac{\permB\, \sigma^{\alpha}\, \rho\, \varI_t}{\sigma^S} \partial_{\alpha\, \hat{\nuB}}\, g_1^I \label{eq:trader_sys_2} \,, \\
        0 &= \partial_{t}\, g_0^I + \frac{(g_1^I)^2}{4 \tempI} - \kappa^{\alpha}\, \alpha\, \partial_{\alpha}\, g_0^I - \theta^B\, \hat{\nuB}\, \partial_{\hat{\nuB}}\, g_0^I + \frac{1}{2} (\sigma^{\alpha})^2 \partial_{\alpha\, \alpha}\, g_0^I \nonumber\\
        &\phantom{{}={}} +\frac{1}{2} \left(\frac{\permB\, \varI_t}{\sigma^S}\right)^2 \partial_{\hat{\nuB}\, \hat{\nuB}}\, g_0^I + \frac{\permB\, \sigma^{\alpha}\, \rho\, \varI_t}{\sigma^S} \partial_{\alpha\, \hat{\nuB}}\, g_0^I \,, \label{eq:trader_sys_3}
\end{align}
with terminal conditions 
\begin{align}
        g_2^I(T) \quad &= -(\beta_0^I + \beta_1^I\, \varI_T)\,, \label{eq:trader_sys_1_tv}\\
        g_1^I(T,\alpha,\hat{\nuB}) &= \qquad 0\,,\qquad \label{eq:trader_sys_2_tv}\\
        g_0^I(T,\alpha,\hat{\nuB}) &= \qquad 0\,.\qquad \label{eq:trader_sys_3_tv}
    \end{align}
The following lemma solves the system of PDEs \eqref{eq:trader_sys_1}-\eqref{eq:trader_sys_2}-\eqref{eq:trader_sys_3}.

\begin{lemma}
\label{lemma:trader_DEs_existence_solutions}
There exist $g_2^I, z_i^I\in C^1([0,T]; \mathbb{R})$ for $i=0,\cdots,8$ such that $g_2^I$ and the following functions
    \begin{align}
        g_1^I(t,\alpha,\hat{\nuB}) &:= z_1^I(t)\, \alpha +z_2^I(t)\, \hat{\nuB} \label{eq:trader_g1}\,, \\
        g_0^I(t,\alpha,\hat{\nuB}) &:= z_3^I(t) + z_4(t)\, \alpha + z_5^I(t)\, \hat{\nuB} + z_6^I(t)\, \alpha\, \hat{\nuB} + z_7^I(t)\, \alpha^2 + z_8^I(t)\, \hat{\nuB}^2 \,,\label{eq:trader_g0}
\end{align}
solve the system of differential equations \eqref{eq:trader_sys_3} with terminal conditions \eqref{eq:trader_sys_1_tv}-\eqref{eq:trader_sys_2_tv}-\eqref{eq:trader_sys_3_tv}.
\end{lemma}

\begin{proof}
    Observe that \eqref{eq:trader_sys_1} is a one-dimensional Riccati equation with quadratic term ${1}/{\tempI}>0$, ``constant'' term $-(\runcostI_t)< 0$, and terminal value $-(\termcostI_t)< 0$. Use Lemma 4.1 in \cite{reid1972rde} to obtain that the equation has a unique solution in $(-\infty,T]$, and that it is negative.
    
    Next, to solve  \eqref{eq:trader_sys_2} and \eqref{eq:trader_sys_3}, one may {work} backwards by substituting the desired form of $g_0$ and $g_1$ into the system of equations 
    to obtain
    \begin{alignat}{2}
      0 &= \alpha &&\left[ 1 - \kappa^{\alpha}\, z_1^I + \frac{g_2^I\, z_1^I}{2\, \tempI} + (z_1^I)' \right]\label{eq:trader_z1}\\
      &\phantom{{}={}} + \hat{\nuB} &&\left[\permB - \theta^B z_2^I + \frac{g_2^I\, z_2^I}{2\, \tempI} + (z_2^I)' \right] \label{eq:trader_z2}\,,\\
      0 &= &&\left[\frac{\permB\, \sigma^{\alpha}\, \rho\, \varI_t}{\sigma^S}\, z_6^I + (\sigma^\alpha)^2\, z_7^I + \left(\frac{\permB\, \varI_t}{\sigma^S}\right)^2 z_8^I + (z_3^I)'\right] \nonumber\\
      &\phantom{{}={}} + \hat{\nuB} &&\left[\frac{g_2^I\, z_2^I}{2\, \tempI} - \theta^B\, z_5^I + (z_5^I)' \right] 
      + \alpha \left[ \frac{g_2^I\, z_1^I}{2\, \tempI} - \kappa^\alpha z_4^I + (z_4^I)' \right] \nonumber\\
      &\phantom{{}={}} + \alpha\, \hat{\nuB} &&\left[ \frac{z_1^I\, z_2^I}{2\, \tempI} - \kappa^\alpha\, z_6^I - \theta^B\, z_6^I + (z_6^I)' \right] 
      + \alpha^2 \left[ \frac{(z_1^I)^2}{4\, \tempI} - 2\, \kappa^\alpha\, z_7^I + (z_7^I)' \right] 
      + \hat{\nuB}^2 \left[ \frac{(z_2^I)^2}{4\, \tempI} - 2\, \theta^B\, z_8^I + (z_8^I)'\right]\,.
    \end{alignat}

Next, define a system of ODEs by setting all the terms in the square bracket above to be zero with $z_i(T)=0$ for $i=0,\cdots,8$. One solves the individual ODEs in the following order:
\begin{equation}
      (z_1^I, z_2^I) \to (z_6^I, z_7^I, z_8^I) \to (z_4^I, z_5^I) \to z_3^I\,.
\end{equation}
In the above order, the ODEs are linear, so there exist $\{z_i^I\}_{i=1,\cdots,8}$ satisfying \eqref{eq:trader_sys_2} and \eqref{eq:trader_sys_3}.
\end{proof}

\subsection{Proof of Theorem \ref{thm: solution Informed problem}}
\begin{proof}
    Observe that \eqref{eq:trader_hjb_solution} satisfies the HJB equation in \eqref{eq:trader_hjb}. Moreover, by substituting $g_i^I$ for $i=0,1,2$ from Lemma \ref{lemma:trader_DEs_existence_solutions}, the solution to the HJB equation is a quadratic polynomial in terms of $\stateI$ with time-dependent  $C^1([0,T])$ coefficients; thus, the solution is bounded by a quadratic growth. By Theorem 3.5.2 (Verification Theorem) in \cite{phamcontrol}, the trader's strategy \eqref{eq:01_informed_strategy_feedback_form} is indeed the optimal Markovian control. The last part of the theorem is as follows: $g_2^I$ is negative from  Lemma \ref{lemma:trader_DEs_existence_solutions}.

    Next, by solving for the ODE \eqref{eq:trader_z2}, we obtain
    \begin{equation}
    \label{eq: trader z2 equation}
        z_2^I(t) = \int_t^T \permB\, \exp \left\{ \int_t^s \frac{g_2^I(u)}{2 \,\tempI} - \theta^B\, \d u \right\}\, \d s \geq 0 \quad \forall t\in [0,T]\,,
    \end{equation}
    which shows that $z_2^I$ is non-negative. 
    The same argument follows for $z_1^I$ by solving \eqref{eq:trader_z1}. Observe that the equation above is linear in $\permB$.
\end{proof}

\section{Derivations for the control problem of the broker}\label{app: derivations broker}

\subsection{Filtering the order flow}\label{app: learning from order flow}

We  compute the dynamics of ${\gamma}_t$ to obtain
\begin{equation}
    \begin{split}
        \d {\gamma}_t &= \quad \left[ f_1'(t)\, \alpha_t - \kappa^\alpha\, f_1(t)\, \alpha_t + f_2'(t)\, \hat{\nuB}_t - \theta^B\, f_2(t)\, \hat{\nuB}_t \right]\, \d t + \left[ \sigma^\alpha\, f_1(t)\, \d W_t^\alpha + \frac{f_2(t)}{(\sigma^S)^2}\, \permB\, \varI_t\,  \left( \d Y_t - \permB\, \hat{\nuB}_t \, \d t \right)\right]\\
        &= \quad \Bigg[\Big(f_1'(t) - \kappa^\alpha\, f_1(t)\Big)\, \alpha_t + \left(f_2'(t) - \theta^B\, f_2(t) - \frac{\permB^2\, \varI_t\, f_2(t)}{\left(\sigma^S\right)^2}\right)\, \hat{\nuB}_t + \frac{\permB^2\, \varI_t\, f_2(t)}{\left(\sigma^S\right)^2}\, \nuB_t \Bigg]\, \d t\\
        &\hspace{8cm} +   \sigma^\alpha\, f_1(t)\, \d W_t^\alpha + \frac{\permB\, \varI_t\, f_2(t)}{\sigma^S}\, \d W_t^S.
    \end{split}
\end{equation}
The broker knows that $$\hat{\nuB}_t = \frac{\nuIstar_t - f_1(t)\, \alpha_t - f_3(t)\, \QI[I*]_t}{f_2(t)} = \frac{\gamma_t - f_1(t)\, \alpha_t}{f_2(t)}\,,$$ so it follows that
\begin{equation}
    \begin{split}
        \d {\gamma}_t &= \left[ \mathfrak{G}^\alpha(t)\, \alpha_t + \mathfrak{G}_0(t)\,{\gamma}_t + \mathfrak{G}_1(t)\, \nuB_t + \mathfrak{G}_2(t)\right]\, \d t\\
        &\qquad\qquad\qquad\qquad + \left[ \mathfrak{G}_3(t)\, \d W_t^\alpha +  \mathfrak{G}_4(t)\, \d W_t^S\right]\,,
    \end{split}
\end{equation}
for functions $\mathfrak{G}_i$ with $i\in\{0,1,2,3,4\}$ and $\mathfrak{G}^\alpha$ given by
\begin{align}
    \mathfrak{G}_0(t) &= -\theta^B - \frac{\permB^2\, \varI_t}{\left(\sigma^S\right)^2} + \frac{f_2'(t)}{f_2(t)}\\
    &= - \frac{\permB^2\, \varI_t}{\left(\sigma^S\right)^2} -\frac{ \permB}{2\, f_2(t)\, \tempI} - \frac{f_3(t)}{2}\,,\\
    \mathfrak{G}^\alpha(t) &= f_1'(t) - f_1(t)\, \left(\kappa^\alpha + \mathfrak{G}_0(t)\right)\\
    &= - \frac{1}{2\, \tempI} - f_1(t)\, \left(\frac{f_3(t)}{2}+ \mathfrak{G}_0(t)\right)\,,\\
    \mathfrak{G}_1(t) &= \frac{\permB^2\, \varI_t\, f_2(t)}{\left(\sigma^S\right)^2}, \quad
    \mathfrak{G}_2(t) = 0\,,\\
    \mathfrak{G}_3(t) &= \sigma^\alpha\, f_1(t),\quad \quad
    \mathfrak{G}_4(t) =\frac{\permB\, \varI_t\, f_2(t)}{\sigma^S}\,.
\end{align}
We define the process $\tilde{Z}_t:=\gamma_t/\mathfrak{G}_5(t)$ with $\mathfrak{G}_5(t):= \sqrt{\mathfrak{G}_3(t)^2+\mathfrak{G}_4(t)^2+2\,\rho\,\mathfrak{G}_3(t)\,\mathfrak{G}_4(t)}$ and obtain
\begin{equation}
\label{eq: broker Z process}
    \begin{split}
        \d \tilde{Z}_t &=\quad \left[ \frac{-\tilde{Z}_t\, \mathfrak{G}_5^{'}(t)}{\mathfrak{G}_5(t)} + \frac{\mathfrak{G}^\alpha(t)}{\mathfrak{G}_5(t)}\, \alpha_t + \frac{\mathfrak{G}_0(t)}{\mathfrak{G}_5(t)}\, \gamma_t + \frac{\mathfrak{G}_1(t)}{\mathfrak{G}_5(t)}\, \nuB_t + \frac{\mathfrak{G}_2(t)}{\mathfrak{G}_5(t)}\right]\, \d t  + \frac{1}{\mathfrak{G}_5(t)}\,\left[ \mathfrak{G}_3(t)\, \d W_t^\alpha +  \mathfrak{G}_4(t)\, \d W_t^S\right]\,.
    \end{split}
\end{equation}
The coefficients in the above expression are well-defined when $t \to T$ and we have that 
\begin{equation}
\label{eq: limit of Z tilde}
    \begin{split}
        \lim_{t\uparrow T} \tilde{Z}_t &= \lim_{t\uparrow T} \frac{f_1(t)}{\mathfrak{G}_5(t)}\, \alpha_t + \frac{f_2(t)}{\mathfrak{G}_5(t)}\, \hat{\nuB}_t= \left( \lim_{t\uparrow T} \frac{f_1(t)}{\mathfrak{G}_5(t)} \right)\, \alpha_T + \left( \lim_{t\uparrow T} \frac{f_2(t)}{\mathfrak{G}_5(t)}\right)\, \hat{\nuB}_T\,,
    \end{split}
\end{equation}
where from  L'Hopital's rule and the differential equations \eqref{eq:trader_z1}--\eqref{eq:trader_z2}, it follows that $\lim_{t\uparrow T} {f_1(t)}/{\mathfrak{G}_5(t)}$ and $\lim_{t\uparrow T} {f_2(t)}/{\mathfrak{G}_5(t)}$ exist and are finite. This implies that $\tilde Z$ is a continuous square-integrable process. Similarly, one also shows that the limits of ${\mathfrak{G}^\alpha(t)}/{\mathfrak{G}_5(t)}$, ${\mathfrak{G}_1(t)}/{\mathfrak{G}_5(t)}$, and ${\mathfrak{G}_2(t)}/{\mathfrak{G}_5(t)}$ exist and are finite as $t\uparrow T$. 
On the other hand, observe that ${\mathfrak{G}_5^{'}(t)}/{\mathfrak{G}_5(t)}$ and ${\mathfrak{G}_0(t)}/{\mathfrak{G}_5(t)}$ blow up as $t\uparrow T$. However, because $(\tilde{Z}_t)_\tT$ and $\big( \int_0^t \alpha_s \,{\mathfrak{G}^\alpha(s)}/{\mathfrak{G}_5(s)} + {\mathfrak{G}_1(s)}/{\mathfrak{G}_5(s)}\, \nuB_s + {\mathfrak{G}_2(s)}/{\mathfrak{G}_5(s)}\, \d s  \big)_\tT$ are linear combinations of square-integrable progressively measurable processes, the process $\big( \int_0^t {-\tilde{Z}_s\, \mathfrak{G}_5^{'}(s)}/{\mathfrak{G}_5(s)} + \gamma_s\, {\mathfrak{G}_0(s)}/{\mathfrak{G}_5(s)}\,  \d s\big)_\tT$ is also square-integrable and progressively measurable. 
Next we define the main process that we use in the filtering framework of this section. The broker observes the process $Z_t$ defined as
\begin{equation}
    \begin{split}
        \d Z_t :&= \d \tilde{Z}_t - \mathfrak{f}_t\, \d t= \mathfrak{G}_7(t)\, \alpha_t\, \d t + \d \tilde{W}_t\,,
    \end{split}
\end{equation}
where $\mathfrak{f}_t = \mathfrak{G}_6(t)\, \tilde{Z}_t + \mathfrak{G}_8(t)\, \gamma_t + \mathfrak{G}_9(t)\, \nuB_t + \mathfrak{G}_{10}(t)$ and
\begin{align}
    \d\, \tilde{W}_t &:= \frac{1}{\mathfrak{G}_5(t)}\,\left[ \mathfrak{G}_3(t)\, \d W_t^\alpha +  \mathfrak{G}_4(t)\, \d W_t^S\right]\,,\label{eq: broker innovation process diffusion}\\
    \mathfrak{G}_6(t) &= -\mathfrak{G}_5^{'}(t)/\mathfrak{G}_5(t)\,,\\
    \mathfrak{G}_7(t) &= \mathfrak{G}^\alpha(t)/\mathfrak{G}_5(t),\quad
    \mathfrak{G}_8(t) = \mathfrak{G}_0(t)/\mathfrak{G}_5(t)\,,\\
    \mathfrak{G}_9(t) &= \mathfrak{G}_1(t)/\mathfrak{G}_5(t),\quad
    \mathfrak{G}_{10}(t) = \mathfrak{G}_2(t)/\mathfrak{G}_5(t)\,.
\end{align}
Use the ODEs \eqref{eq:trader_z1}, \eqref{eq:trader_z2}, and \eqref{eq: varI}, to rewrite $\mathfrak{G}_6(t)$ as
\begin{equation}
    \mathfrak{G}_6(t) = - \frac{\mathfrak{G}_3^{'}(t)[\mathfrak{G}_3(t)+\rho\,\mathfrak{G}_4(t)] + \mathfrak{G}_4^{'}(t)[\rho\, \mathfrak{G}_3(t)+\mathfrak{G}_4(t)]}{\mathfrak{G}_5(t)^2}\,,
\end{equation}
where
\begin{align}
    \mathfrak{G}_3^{'}(t) &= \sigma^\alpha \left[ -\frac{1}{2\, \tempI} + \kappa^\alpha\, f_1(t) - \frac{f_3(t)\, f_1(t)}{2}  \right]\,,\\
    \mathfrak{G}_4^{'}(t) &= \frac{\permB}{\sigma^S} \left[ \varI_t\, \left( -\frac{\permB}{2\, \tempI} + \theta^B\, f_2(t) - \frac{f_3(t)\, f_2(t)}{2} \right) + f_2(t)\, \left( (\sigma^B)^2 - 2\, \theta^B\, \varI_t - \frac{\permB^2\, (\varI_t)^2}{(\sigma^S)^2} \right) \right]\,.
\end{align}

\subsection{Proof of Proposition \ref{prop: alternative filter from trader's speed dynamics}}\label{app:proof of proposition}

We need the following lemmas to derive the filtering equations from the dynamics of the informed trader's speed. Recall that the process $Z$ (which is observable) follows the equation
\begin{equation}
    \begin{split}
        \d Z_t &= \mathfrak{G}_7(t)\, \alpha_t\, \d t + \d \tilde{W}_t
    \end{split}
\end{equation}
where the signal $\alpha$ is latent and $\tilde{W}$ is given in the lemma below.
\begin{lemma}
\label{lemma: new brownian motion}
    The process below
    \begin{equation}
    \label{eq: broker innovation process diffusion 1}
        \d\tilde{W}_t:= \frac{1}{\mathfrak{G}_5(t)}\,\left[ \mathfrak{G}_3(t)\, \d W_t^\alpha +  \mathfrak{G}_4(t)\, \d W_t^S\right]
    \end{equation}
    is a $(\Pb, (\mathcal{F}_t)_t)$-Brownian motion. Furthermore, the process below (called the innovation process)
    \begin{equation}
        I_t := Z_t - \int_0^t \E \left[ \mathfrak{G}_7(s)\, \alpha_s \, \Big| \,\mathcal{Z}_s \right]\, \d s
    \end{equation}
    is a $(\Pb, (\mathcal{Z}_t))$-Brownian motion.
\end{lemma}
\begin{proof}
    The first part of the lemma is obvious as the Ito integrand are bounded by $1$ and the quadratic variation is $\d t$. Next, we will prove that it is a $(\Pb, (\mathcal{Z}_t))$-martingale. Let $0\leq s\leq t\leq T$. Then
    \begin{equation}
    \label{eq: broker innovation process}
        \begin{split}
            \E \left[ I_t - I_s\, |\, \mathcal{Z}_s\right] &= \E \left[Z_t - Z_s -\int_s^t \E \left[ \mathfrak{G}_7(u)\, \alpha_u\, \Big| \,\mathcal{Z}_u \right]\, \d u\, \Big|\, \mathcal{Z}_s\right]\\
            &=  \E \left[Z_t - Z_s\, | \, \mathcal{Z}_s\right] - \E \left[\int_s^t \mathfrak{G}_7(u)\, \alpha_u\, \d u\, \Big| \,\mathcal{Z}_s\right]\,.
        \end{split}
    \end{equation}
    Note that the diffusion terms in \eqref{eq: broker innovation process diffusion 1} are bounded by $1$, so that 
    $$\E \left[ Z_t - Z_s \, |\, \mathcal{F}_s\right] = \E \left[ \int_s^t  \mathfrak{G}_7(u)\, \alpha_u\, \d u \Big|\, \mathcal{F}_s\right]\,.$$
    But $\mathcal{Z}_s\subset \mathcal{F}_s$, so that $\E [Z_t - Z_s \, | \mathcal{Z}_s] = \E [ \E [Z_t - Z_s\, |\, \mathcal{F}_s] \, |\, \mathcal{Z}_s]$, and we obtain \eqref{eq: broker innovation process} equals zero.

    Next, it can be easily calculated that $\langle I\rangle_t = t$. By Levy Characterisation, $I$ is a $(\Pb, (\mathcal{Z}_t)_t)$-Brownian motion.
\end{proof}
\begin{lemma}
\label{lemma: martingale representation}
    The processes 
    \begin{align}
        M_t &= \E \left[-\int_0^t \kappa^\alpha\, \alpha_u\, \d u + \sigma^\alpha\, W_t^\alpha \, \big|\, \mathcal{Z}_t \right] + \int_0^t \kappa^\alpha\, \E\left[ \alpha_u\, |\, \mathcal{Z}_u\right]\, \d u\,,\\
        N_t &= \E \left[-\int_0^t \kappa^\alpha\, \alpha_u^2\, \d u + \int_0^t \sigma^\alpha\, \alpha_u\, \d W_u^\alpha \, \big|\, \mathcal{Z}_t \right] + \int_0^t \kappa^\alpha\, \E\left[ \alpha_u^2\, |\, \mathcal{Z}_u\right]\, \d u\,
    \end{align}
    are $(\Pb, (\mathcal{Z}_t)_\tT)$-martingales.
\end{lemma}
\begin{proof}
    Observe that for $s\leq t$, and the fact that $W^\alpha$ is a $(\Pb, (\mathcal{F}_t)_\tT)$-Brownian motion, we may obtain the following by using the tower property.
    \begin{align}
        \E \left[ M_t - M_s\, \big|\, \mathcal{Z}_s\right] &= \E \left[-\int_s^t \kappa^\alpha\, \alpha_u\, \d u \, \big|\, \mathcal{Z}_s \right] + \E \left[ \int_s^t \kappa^\alpha\, \E\left[ \alpha_u\, |\, \mathcal{Z}_u\right]\, \d u\, \big|\, \mathcal{Z}_s\right] = 0\,.
    \end{align}
    Additionally, $\E \left[\int_0^T \alpha_u^2\, \d u\right] < \infty$ because $\E\left[\alpha_u^2\right] = \frac{(\sigma^\alpha)^2}{2\, \kappa^\alpha} \exp(-2\, \kappa^\alpha\, u)$. Thus, $\left( \int_0^t \alpha_u\, \d W_u^\alpha \right)_\tT$ is a $(\Pb, (\mathcal{F}_t)_\tT)$-martingale. Using the same trick as before,
    \begin{align}
        \E \left[ N_t - N_s\, \big|\, \mathcal{Z}_s\right] &= \E \left[-\int_s^t \kappa^\alpha\, \alpha_u^2\, \d u \, \big|\, \mathcal{Z}_s \right] + \E \left[ \int_s^t \kappa^\alpha\, \E\left[ \alpha_u^2\, |\, \mathcal{Z}_u\right]\, \d u\, \big|\, \mathcal{Z}_s\right] = 0\,.
    \end{align}
\end{proof}
\begin{proof}[Proof for Proposition \ref{prop: alternative filter from trader's speed dynamics}]
    From Lemma \ref{lemma: martingale representation}, $M$ is a $(\Pb, (\mathcal{Z}_t)_\tT)$-martingale. Because and $I$ is a $(\Pb, (\mathcal{Z}_t)_\tT)$-Brownian motion, by the martingale representation theorem, there exists a $(\mathcal{Z}_t)_\tT$-adapted process $r^M$ such that $M_t = \int_0^t r_u^M\, \d I_u$. We compute
    \begin{equation}
    \label{eq: alpha estimate 1}
        \hat{\alpha}^{\text{alt}}_t = \E \left[-\int_0^t \kappa^\alpha\, \alpha_u\, \d u + \sigma^\alpha\, W_t^\alpha \, \big|\, \mathcal{Z}_t \right] = M_t - \int_0^t \kappa^\alpha\, \hat{\alpha}^{\text{alt}}_u\, \d u = \int_0^t r_u^M \, \d I_u\ - \int_0^t \kappa^\alpha\, \hat{\alpha}^{\text{alt}}_u\, \d u\,.
    \end{equation}
    By applying Ito's formula to $(\hat{\alpha}^{\text{alt}}_t)^2$, we obtain
    \begin{equation}
        \d (\hat{\alpha}^{\text{alt}}_t)^2 = \left[ -2\, \kappa^\alpha\, (\hat{\alpha}^{\text{alt}}_t)^2 + (r_t^M)^2\right]\, \d t + 2\, r_t^M\, \hat{\alpha}^{\text{alt}}_t\, \d I_t\,.
    \end{equation}
    Similarly for $N_t$, there exists a $(\mathcal{Z}_t)_\tT$-adapted process $r^N$ such that $N_t = \int_0^t r_u^N\, \d I_u$. Now, we apply Ito's formula to $\alpha_t^2$ and take its conditional expectation with respect to $\mathcal{Z}_t$.
    \begin{equation}
        \begin{split}
            \E \left[ \alpha_t^2 \, \big|\, \mathcal{Z}_t\right] &= \int_0^t \Big\{ (\sigma^\alpha)^2 - 2\, \kappa^\alpha\, \E \left[ \alpha_u^2 \, \big|\, \mathcal{Z}_t\right] \Big\}\, \d u + \sigma^\alpha\, \E \left[ \int_0^t 2\, \alpha_u\, \d W_u^\alpha \, \big|\, \mathcal{Z}_t\right]\\
            &= \int_0^t 2\, r_u^N\, \d I_u + \int_0^t \Big\{(\sigma^\alpha)^2 - 2\, \kappa^\alpha\, \E \left[ \alpha_u^2 \, \big|\, \mathcal{Z}_u\right]\Big\}\, \d u\,.
        \end{split}
    \end{equation}
    From Lemma \ref{lemma: new brownian motion}, $\tilde{W}$ is a $(\Pb, (\mathcal{F}_t)_\tT)$-Brownian motion. Because the signal $\alpha$ is a $\Pb$-Gaussian process, then $Z$ is also a $\Pb$-Gaussian process. Thus, the projection $\alpha_t$ into $\mathcal{Z}_t$ is also Gaussian. Thus, $\alpha_t - \hat{\alpha}^{\text{alt}}_t$ is independent of $\mathcal{Z}_t$, so that $\varB[\text{alt},B]_t$ must be a deterministic quantity. As such, we should have $r_t^M\, \hat{\alpha}^{\text{alt}}_t = r_u^N$, and
    \begin{equation}
    \label{eq: alpha conditional variance 1}
        \d \varB[\text{alt},B]_t = \Big\{(\sigma^\alpha)^2 - 2\, \kappa^\alpha\, \varB[\text{alt},B]_t - (r_t^M)^2\Big\}\, \d t\,.
    \end{equation}

    Now, we wish to determine the process $r^M$. Define $\epsilon_t := \exp \left( i \, \int_0^t r(u)\, \d Z_u + \frac{1}{2}\int_0^t |r(u)|^2\, du\right)$ where $r: [0,T]\to \Real$ is measurable and bounded. By Ito's formula, $\d \epsilon_t = i\, \epsilon_t\, r(t)\, \d Z_t$. We calculate
    \begin{align}
        \d \left( \epsilon\, \alpha\right)_t &= \epsilon_t\, \d \alpha_t + \alpha_t\, \d \epsilon_t + \d\langle \epsilon, \alpha \rangle_t\\
        &= -\kappa^\alpha\, \epsilon_t\, \alpha_t\, \d t + \sigma^\alpha\, \epsilon_t\, \d W_t^\alpha + i\, \epsilon_t\, \alpha_t\, r(t)\, \d Z_t + i\, \sigma^\alpha\, \epsilon_t\, r(t)\, \mathfrak{K}(t)\, \d t\label{eq: finding rM 1}\,,\\
        \d \left(\epsilon\, \hat{\alpha}^{\text{alt}}\right)_t &= \epsilon_t\, \d \hat{\alpha}^{\text{alt}}_t + \hat{\alpha_t}^{\text{alt}}\, \d \epsilon_t + \d\langle \epsilon, \hat{\alpha}^{\text{alt}} \rangle_t\\
        &= -\kappa^\alpha\, \epsilon_t\, \hat{\alpha}^{\text{alt}}_t\, \d t + \epsilon_t\, r^M_t\, \d I_t + i\, \epsilon_t\, r(t)\, \hat{\alpha}^{\text{alt}}_t\, \d Z_t + i\, \epsilon_t\, r(t)\, r_t^M\, \d t\label{eq: finding rM 2}\,,
    \end{align}
    where 
    $$ \mathfrak{K}(t) = \frac{\left(\mathfrak{G}_3(t) + \rho\, \mathfrak{G}_4(t)\right) }{\mathfrak{G}_5(t)}\,.
    $$
    Note that by applying the tower property, $\E \left[ \epsilon_t (\hat{\alpha}^{\text{alt}}_t - \alpha_t)\right] = 0$. By subtracting \eqref{eq: finding rM 1} from \eqref{eq: finding rM 2} and taking the total expectation, we have
    \begin{equation}
    \label{eq: finding Ito integrand}
        \begin{split}
            0 &=  \E \left[\int_0^t  i\, \epsilon_u\, r(u)\, \left(\hat{\alpha}^{\text{alt}}_u - \alpha_u \right)\, \d Z_u \right] + \int_0^t i\, r(u)\, \E \left[ \epsilon_u\, (r_u^M - \sigma^\alpha\, \mathfrak{K}(u))\right]\, \d u\\
            &= \int_0^t i\, r(u)\, \E \left[ \epsilon_u\, \left\{r_u^M - \sigma^\alpha\, \mathfrak{K}(u) +  (\hat{\alpha}^{\text{alt}}_u - \alpha_u)\, \mathfrak{G}_7(u)\, \alpha_u\right\}\right]\, \d u +  \E \left[\int_0^t  i\, \epsilon_u\, r(u)\, \left(\hat{\alpha}^{\text{alt}}_u - \alpha_u \right)\, \d \tilde{W}_u \right]\,.
        \end{split}
    \end{equation}
    We know that $\varB[\text{alt},B]_t \leq \text{Var}(\alpha_t) = \frac{(\sigma^\alpha)^2}{2\, \kappa^\alpha}\, \exp \left(-2\, \kappa^\alpha\, t \right)$, so that $\E \left[ \int_0^T \left(\hat{\alpha}^{\text{alt}}_u - \alpha_u \right)^2\, \d u\right] < \infty$. Thus, the rightmost term in \eqref{eq: finding Ito integrand} equals zero. We obtain
    $$ 0 = \int_0^t i\, r(u)\, \E \left[ \epsilon_u\, \left\{r_u^M - \sigma^\alpha\, \mathfrak{K}(u) +  (\hat{\alpha}^{\text{alt}}_u - \alpha_u)\, \mathfrak{G}_7(u)\, \alpha_u\right\}\right]\, \d u, \quad \forall \tT,$$
    and consequently,
    $$ \E \left[ \epsilon_t\, \left\{r_t^M - \sigma^\alpha\, \mathfrak{K}(t) +  (\hat{\alpha}^{\text{alt}}_t - \alpha_t)\, \mathfrak{G}_7(t)\, \alpha_t\right\}\right] = 0, \quad \forall \tT.$$
    This equation applies for all $r\in L^\infty([0,T];\Real)$. Using Lemma B.39 in \cite{filtering_alain_crisan}, we have $\E \left[r_t^M - \sigma^\alpha\, \mathfrak{K}(t) +  (\hat{\alpha}^{\text{alt}}_t - \alpha_t)\, \mathfrak{G}_7(t)\, \alpha_t\, \big| \, \mathcal{Z}_t\right] = 0$ almost surely. Thus, we obtain
    \begin{equation}
        r_t^M = \sigma^\alpha\, \mathfrak{K}(t) + \mathfrak{G}_7(t)\, \varB[\text{alt},B]_t\,.
    \end{equation}
    Substituting this into \eqref{eq: alpha estimate 1} and \eqref{eq: alpha conditional variance 1}, we obtain \eqref{eq: alpha mean estimate} and \eqref{eq: alpha var estimate} in Proposition \ref{prop: alternative filter from trader's speed dynamics}.
\end{proof}

\subsection{Hamilton–Jacobi–Bellman equation} 
The HJB equation associated with the value function $H$ in \eqref{eq: broker control problem} is 
\begin{align}
0 = \partial_t H + \sup_{\nuB} \bigg\{&
   \left( \permB\, \nuB + \hat{\alpha} \right)\,\partial_s H +\frac{1}{2}\, \left(\sigma^S\right)^2\, \partial_{ss} H \\
  & - \kappa^\alpha\,\hat{\alpha}\, \partial_{\hat{\alpha}} H + \frac{1}{2}\,\left(\frac{\varB_t + \rho\,\sigma^S\,\sigma^\alpha}{\sigma^S} \right)^2\,\partial_{\hat{\alpha},\hat{\alpha}} H + \left(\varB_t + \rho\,\sigma^S\,\sigma^\alpha\right)\, \partial_{\hat{\alpha},s} H\\
    &  -\nuB\, \left(  s + \tempB\, \nuB\right)\,\partial_{x^B} H + \nuIstar\, \left(  s + \tempI\, \nuIstar\right) \,\partial_{x^B} H + \nuU\, \left(  s + \tempU\, \nuU\right) \,\partial_{x^B} H \\
& + \left(\nuB-\nuIstar-\nuU\right)\, \partial_{q^B} H\\
& + \left(\nuIstar\right)\, \partial_{q^I} H  - \kappa^U\,\nuU\,\partial_\nuU H + \frac12 \, \left(\sigma^U\right)^2 \, \partial_{\nuU\nuU} H  -\left(\rho_0^B + \rho_1^B\,\varB_t\right) \left(q^B\right)^2
\bigg\}\,, \label{eq: broker HJB}
\end{align}
with terminal condition $$H(T,\stateB) = x^B + q^B\,s - \left(\beta^B_0 +\beta^B_1\,\varB_T\right)\,\left(q^B\right)^2\,,$$ 
where 
\begin{equation}\label{eq: nuIstar for HJB broker}
    \nuIstar = \frac{z_1^I(t)\,\hat{\alpha} + z_2^I(t)\,\nuB  + 2\,g_2^I (t)\,q^I }{2\,\tempI} \,.
\end{equation}
Next, define $$f_1(t) = z_1^I(t)/(2\, \tempI)\,, \quad f_2(t) = z_2^I(t)/(2\,\tempI)\,,\quad \text{and } f_3(t) =g_2^I(t)/\tempI\,,$$ 
and employ the ansatz 
$$H = x^B + h(t,s,q^B,\hat{\alpha},\nuU,q^I)$$ 
to obtain the following  PDE in $h$
\begin{equation} \label{eq:broker pde h}
     \begin{split}
        0 &= \partial_t h + \partial_s h\,\hat{\alpha} +  \nuU \left(  s + \tempU\, \nuU\right) - \partial_{q^B} h\, \nuU + \partial_{\hat{\alpha}} h\left(-\kappa^\alpha\hat{\alpha}\right) + \partial_{\nuU}h\left(-\kappa^U\,\nuU\right)\\
        &\phantom{{}={}} + \frac{1}{2}\partial_{ss} h\left(\sigma^S\right)^2 + \frac{1}{2} \partial_{\hat{\alpha}\hat{\alpha}}h \left(\frac{\varB_t + \rho\,\sigma^S \sigma^\alpha}{\sigma^S}\right)^2 + \frac{1}{2} \partial_{\nuU, \nuU}h\left(\sigma^U\right)^2\\
		&\phantom{{}={}} + \partial_{s\hat{\alpha}}h\left[\varB_t + \rho\,\sigma^S\,\sigma^\alpha\right] - (\rho_0^B + \rho_1^B\,\varB_t)\,\left(q^B\right)^2\\
        &\phantom{{}={}} + \left[ \hat{\alpha}\, f_1 + f_3\, q^I  \right]\partial_{q^I}h - \left[ \hat{\alpha}\, f_1 + f_3\, q^I \right]\, \partial_{q^B}h\\
        &\phantom{{}={}} + \Bigg[  \hat{\alpha}^2\, f_1^2\, \tempI  + 2\, \hat{\alpha}\, f_1\, f_3\, \tempI\, q^I + f_3^2\, \tempI\, \left(q^I\right)^2  + \hat{\alpha}\, f_1\, s + f_3\, q^I\, s \Bigg]\\
        &\phantom{{}={}} + \sup_{\nuB} \Bigg\{ -\left( \nuB \right)^2 \left( \tempB - f_2^2\, \tempI \right) + \nuB\bigg[ f_2\, \partial_{q^I}h + \left( 1 - f_2 \right)\, \partial_{q^B}h \\
        &\phantom{{}={}}\qquad  + \left(  2\, \hat{\alpha}\, f_1\, f_2\, \tempI + 2\, f_2\, f_3\, \tempI\, q^I - s + f_2\, s \right)\, + \permB\, \partial_s h \bigg] \Bigg\}\,,
     \end{split}
\end{equation}
with terminal condition $$h(T,s,q^B, \hat{\alpha}, \nuU, q^I) = q^B\,s - \left(\beta^B_0 +\beta^B_1\,\varB_T\right)\,\left(q^B\right)^2\,.$$
The supremum in \eqref{eq:broker pde h} is attained at 
\begin{equation}
    \nuBstar = \frac{f_2\, \partial_{q^I}h + \left( 1 - f_2 \right)\, \partial_{q^B}h + \left(  2\, \hat{\alpha}\, f_1\, f_2\, \tempI + 2\, f_2\, f_3\, \tempI\, q^I - s + f_2\, s \right) + \permB\, \partial_s h}{2\, \left(\tempB - f_2^2\, \tempI\right)}\,,
\end{equation}
because $\left(\tempB - f_2^2\, \tempI\right) > 0$, so the function $h$ satisfies the PDE
\begin{equation}\label{eq:broker hjb h 2}
     \begin{split}
        0 &= \partial_t h + \partial_s h\, \hat{\alpha} + \nuU\,\left(  s + \tempU\, \nuU\right) - \partial_{q^B} h\,\nuU + \partial_{\hat{\alpha}}h\left[-\kappa^\alpha\,\hat{\alpha}\right] + \partial_{\nuU}h\left[-\kappa^U\,\nuU\right]\\
        &\phantom{{}={}} + \frac{1}{2}\partial_{ss}h\left(\sigma^S\right)^2 + \frac{1}{2} \partial_{\hat{\alpha}\hat{\alpha}}h \left(\frac{\varB_t + \rho\,\sigma^S\,\sigma^\alpha}{\sigma^S}\right)^2 + \frac{1}{2} \partial_{\nuU, \nuU}h\left(\sigma^U\right)^2\\
		&\phantom{{}={}} + \partial_{s\hat{\alpha}}h \left[\varB_t + \rho\,\sigma^S \sigma^\alpha\right] - (\rho_0^B + \rho_1^B \,\varB_t)\,q^2\\
        &\phantom{{}={}} + \left[\hat{\alpha}\, f_1 + f_3\, q^I  \right]\partial_{q^I}h - \left[ \hat{\alpha}\, f_1 + f_3\, q^I \right]\, \partial_{q^B}h\\
        &\phantom{{}={}} + \Bigg[ \hat{\alpha}^2\, f_1^2\, \tempI + 2\, \hat{\alpha}\,f_1\, f_3\, \tempI\, q^I + f_3^2\, \tempI\, \left(q^I\right)^2  + \hat{\alpha}\, f_1\, s + f_3\, q^I\, s \Bigg]\\
        &\phantom{{}={}} + \frac{\left[f_2\, 
        \partial_{q^I}h + \left( 1 - f_2 \right)\, \partial_{q^B}h + \left( 2\, \hat{\alpha}\, f_1\, f_2\, \tempI + 2\, f_2\, f_3\, \tempI\, q^I - s + f_2\, s \right) + \permB\, \partial_s h\right]^2}{4\, \left(\tempB - f_2^2\, \tempI\right)}\,.
     \end{split}
\end{equation}

\paragraph{\textbf{The HJB equation in matrix form}} 
To study the problem further, we show that the PDE \eqref{eq:broker hjb h 2} is a system of matrix differential equations involving a matrix Riccati ODE. Employ the ansatz
\begin{equation}
    h(t, s, \tilde{\stateB}) = q^B\, s + G_0 + 2\, G_1\, \tilde{\stateB}^\top + \tilde{\stateB}\, G_2\, \tilde{\stateB}^\top
\end{equation}
where $\tilde{\stateB} = \begin{pmatrix} q^B & \hat{\alpha} & \nuU & q^I \end{pmatrix}$, $G_0:[0,T]\to \Real$, $G_1:[0,T]\to \Real^{1\times 4}$, and $G_2:[0,T]\to \Real_\text{symmetric}^{4\times 4}$ are functions of time, and write the HJB equation \eqref{eq:broker hjb h 2} as
\begin{alignat}{2}
    0 &= &&\Bigg[ G_0' + 4\, G_1\, P_8^\top\, P_8\, G_1^\top + \left(\frac{\varB_t + \rho\,\sigma^S\,\sigma^\alpha}{\sigma^S}\right)^2\, e_2\, G_2\, e_2^\top + (\sigmaU)^2\, e_3\, G_2\, e_3^\top\Bigg]\\
    &\phantom{{}={}} + 2\, &&\Bigg[G_1' + G_1\, P_2^\top + 2\, G_1\, P_8^\top\, P_7  + 4\, G_1\, P_8^\top\, P_8\, G_2 \Bigg]\, \tilde{\stateB}^\top\\
    &\phantom{{}={}} + \tilde{\stateB}\, &&\Bigg[G_2' + 2\, P_2\, G_2 + P_7^\top\, P_7 + 4\, G_2^\top\, P_8^\top\, P_8\, G_2 + 4\, G_2^\top\, P_8^\top\, P_7 + P_5 \Bigg]\, \tilde{\stateB}^\top\,,
\end{alignat}
where
\begin{align}
    P_2 &=
    \left(
    \begin{array}{ccccc}
     0 & 0 & 0 & 0  \\
     -{f_1} & -{\kappa^\alpha} & 0 & {f_1}  \\
     -1 & 0 & -{\kappa^U} & 0  \\
     -{f_3} & 0 & 0 & {f_3}  \\
    \end{array}
    \right)\,,
    \\[0.3em]
    P_5 &= 
    \left(
    \begin{array}{ccccc}
     -{\varB} {\rho_1^B}-{\rho_0^B} & 1/2 & 0 & 0 \\
     1/2 & {f_1}^2 {\tempI} & 0 & {f_1} {f_3} {\tempI} \\
     0 & 0 & {\tempU} & 0 & \\
     0 & {f_1} {f_3} {\tempI} & 0 & {f_3}^2 {\tempI} & \\
    \end{array}
    \right)\,,\\
    P_7 &=
    \left(
    \begin{array}{ccccc}
     \frac{\permB}{2 \sqrt{{\tempB}-{f_2}^2 {\tempI}}} & \frac{{f_1} {f_2} {\tempI}}{\sqrt{{\tempB}-{f_2}^2 {\tempI}}} & 0 & \frac{{f_2} {f_3} {\tempI}}{\sqrt{{\tempB}-{f_2}^2 {\tempI}}}  \\
    \end{array}
    \right)\,,\label{eq: P7}\\[1em]
    P_8 &=
    \left(
    \begin{array}{ccccc}
     \frac{1-{f_2}}{2 \sqrt{{\tempB}-{f_2}^2 {\tempI}}} & 0 & 0 & \frac{{f_2}}{2 \sqrt{{\tempB}-{f_2}^2 {\tempI}}} \\
    \end{array}
    \right)\,,\label{eq: P8}
\end{align}
and $\{e^i\}$ is the canonical basis of $\mathbb R^4$. Recall that $\tilde{\stateB}\, M\, \tilde{\stateB}^\top = \tilde{\stateB}\, M^\top\, \tilde{\stateB}^\top$ for any four-by-four matrix $M$ and define
\begin{equation}
    P_9 = 2\, P_8^\top\, P_7 + P_2^\top\,, 
\end{equation}
to obtain the matrix differential equations
\begin{alignat}{2}
    0 &= &&\, G_0' + 4\, G_1\, P_8^\top\, P_8\, G_1^\top + \left(\frac{\varB_t + \rho\,\sigma^S\,\sigma^\alpha}{\sigma^S}\right)^2\, e_2\, G_2\, e_2^\top + (\sigmaU)^2\, e_3\, G_2\, e_3^\top \,, \label{eq: broker G0}\\
    0 &=&&\,G_1' + G_1\, P_2^\top + 2\, G_1\, P_8^\top\, P_7 + 4\, G_1\, P_8^\top\, P_8\, G_2 \,, \label{eq: broker G1}\\
    0 &= &&\,G_2' + P_7^\top\, P_7 + 4\, G_2^\top\, P_8^\top\, P_8\, G_2 + G_2^\top\, P_9 + P_9^\top\, G_2 + P_5\,, \label{eq: broker riccati}
\end{alignat}
with terminal conditions $$G_0 = 0\,, \quad G_1=0\,,$$
and
\begin{equation}\label{eq:term cond G2}
    G_2(T) = \begin{pmatrix}
        -(\termcostB_T) & 0 & 0 & 0\\
        0 & 0 & 0 & 0\\
        0 & 0 & 0 & 0\\
        0 & 0 & 0 & 0
    \end{pmatrix}\,.
\end{equation}

\subsection{Proof of Proposition \ref{prop: freiling riccati existence}}

\begin{proof}
    The  matrix Riccati equation in \eqref{eq: broker riccati} can be simplified because the components of $P_8^\top\, P_8$ are zero except in the corners, and the first and last columns of $P_9$ take the form $\begin{pmatrix} \cdots & 0 & 0 & \cdots \end{pmatrix}^\top$. The simplified differential equation satisfied by $\tilde{G}_2$  is
    \begin{equation}
    \label{eq: broker riccati reduced}
        \begin{split}
            0 = \tilde{G}_2' + 4\, \tilde{G}_2^\top\, U\, \tilde{G}_2 + \tilde{G}_2^\top\, V + V^\top\, \tilde{G}_2 + B\,, \quad
            \tilde{G}_2(T) = \begin{pmatrix}
            -(\termcostB_T) & 0 \\
            0 & 0
        \end{pmatrix}\,,
        \end{split}
    \end{equation}
    with
    \begin{equation}
    \begin{split}    
        \tilde{G}_2 &=
        \begin{pmatrix}
            (G_2)_{11} & (G_2)_{14}\\
            (G_2)_{14} & (G_2)_{44}
        \end{pmatrix}\,,
        \\
        U &= \frac{1}{\tempB-f_2^2 \tempI}\,
        \begin{pmatrix}
            (1-f_2)^2 & f_2\,(1-f_2)\\
            f_2\,(1-f_2) & f_2^2
        \end{pmatrix}
        = \begin{pmatrix}
            1-f_2\\
            f_2
        \end{pmatrix}\,
        \left(
            \tempB-f_2^2 \tempI
        \right)^{-1}\,
        \begin{pmatrix}
            1-f_2 & f_2
        \end{pmatrix}\,,
        \\
        V &= \frac{1}{\tempB-f_2^2 \tempI}\,
        \begin{pmatrix}
            \frac{\permB\, (1-f_2)}{2} & -f_3\, (\tempB - f_2\, \tempI)\\
            \frac{\permB\, f_2}{2} & \tempB\, f_3
        \end{pmatrix}\,,
        \\
        B &= \frac{1}{\tempB-f_2^2 \tempI}\,
        \begin{pmatrix}
            \frac{\permB^2}{4} - \left(\constrB \right)\left(\runcostB_t\right) & \frac{\permB\,f_2\,f_3\, \tempI}{2}\\
            \frac{\permB\,f_2\,f_3\, \tempI}{2} & f_3^2\, \tempB\, \tempI
        \end{pmatrix}\,.
    \end{split}
    \end{equation}
Recall that $z_2^I$ satisfies \eqref{eq: trader z2 equation}, so 
    $f_2(t) = \permB\, \int_t^T \exp \left\{ \int_t^s \frac{f_3(u)}{2} - \theta^B\, \d u \right\}\, \d s$, $f_2 = 0$ whenever $\permB=0$, and $f_2$ is continuous in $\permB$.  
\noindent Next, use Theorem 2.3 in \cite{freiling2000non} by setting $$C = \begin{pmatrix} 0 & 0\\ 0 & 1 \end{pmatrix}\,, \qquad \displaystyle D = -\begin{pmatrix} 1 & 0\\ 0 & 1\end{pmatrix}\,,$$ and consider the following matrices:
    \begin{align}
        B_{11} &= V\,,\quad B_{12} = U\,,\quad B_{21} = -B^\top\,,\quad B_{22} = -V^\top\,,\\
        E &= C + D\, \tilde{G}_2(T) + \left(D\,\tilde{G}_2(T)\right)^\top\,,\\
        L(t) &= \begin{pmatrix}
            C\, B_{11} + D\, B_{21} & C\, B_{12} + B_{11}^\top\, D + D\, B_{22}\\
            {0} & B_{12}^\top\, D
        \end{pmatrix}\,\\
        &= \begin{pmatrix}
            C\, V + B^\top & C\, U\\
            {0} & -U^\top
        \end{pmatrix}\,.
    \end{align}
    It follows that
    \begin{equation}
            E = \begin{pmatrix}
            2\, (\termcostB_t) & 0\\
            0 & 1
            \end{pmatrix} > 0\,,\quad \text{and}\quad 
            L + L^\top = \begin{pmatrix}
            C\, V + V^\top\, C + 2\,B & C\, U\\
            U\, C & -2\,U
            \end{pmatrix}\,.
    \end{equation}
    Observe that $\ell(t) := \begin{pmatrix} -f_2(t) & 1 - f_2(t)\end{pmatrix}^\top$ is orthogonal to $\begin{pmatrix} 1 - f_2(t) & f_2(t)\end{pmatrix}^\top$, so that $U(t)\, \ell(t) = 0$. Because the last two columns of $L + L^\top$, which are comprised of $C\, U$ and $-2\, U$, are linearly dependent, then $\det(L(t) + L(t)^\top) = 0$ for all $t\in[0,T]$. Thus, $L + L^\top$ has at least one zero eigenvalue. If $\permB = 0$, then $f_2 = 0$. It turns out that for this parameter, $L + L^\top$ is a diagonal matrix with eigenvalues $\left\{ -2\, (\runcostB_t),\, 2\, f_3(t)\, (1 + \tempI\, f_3(t)),\, -2/\tempB,\, 0 \right\}$.  Given that $g_2^I<0$, it follows that  $f_3(t)< 0$, and we assumed that $1 + \tempI\, f_3(t) > 0$ for all $t\in [0,T]$, thus,  the first three eigenvalues are negative. Next, if we perturb $\permB$ slightly, the eigenvalues  only change slightly because of their continuity in $\permB$; thus, the first three eigenvalues stay negative. Furthermore, given that $\det(L + L^\top) = 0$, the last eigenvalue remains zero. Consequently, there exists $\delta>0$ such that for all $\permB\in [0,\delta)$, the matrix $L(t) + L(t)^\top$ is negative semidefinite for all $t\in [0,T]$. From Theorem 2.3 of \cite{freiling2000non}, the solution for the matrix Riccati equation \eqref{eq: broker riccati reduced} exists for all $t\in [0, T]$. 
\end{proof}

\subsection{{Proof of Theorem \ref{thm: solution problem broker}}}\label{app: proof of theorem broker}
\begin{proof}
By Proposition \ref{prop: freiling riccati existence} there is a unique solution to the matrix Riccati differential equation \eqref{eq: broker riccati}. Armed with $G_2$,  we  find the unique solution $G_0$ to the linear differential equation \eqref{eq: broker G0}. With $G_0$ and $G_2$, we obtain the following solution to the HJB equation \eqref{eq: broker HJB}:
    \begin{equation}
         H^B(t, \stateB) = x^B + q^B\, s + G_0(t) + \tilde{\stateB}\, G_2(t)\, \tilde{\stateB}^\top\,.
    \end{equation}
 We observe that $H^B$ is a quadratic polynomial in $\stateB$ with time-dependent $C^1([0,T], \Real)$ coefficients; thus, it is bounded by a quadratic growth. By Theorem 3.5.2 in \cite{phamcontrol}, the broker's strategy \eqref{eq:broker_optimal_trading_rate_feedback} is indeed her optimal Markovian control.
\end{proof}

\subsection{Existence and uniqueness details}\label{app: details existence uniqueness eigenvalues}

The parameters we use satisfy the conditions for existence and uniqueness of the Riccati equation.
Figure \ref{fig:freiling_riccati_eigenvals} shows the evolution of the first three eigenvalues of $L(t) + L(t)^\top$. Observe that the first three eigenvalues are negative, while the fourth eigenvalue is  zero, as required in the proof of Proposition \ref{prop: freiling riccati existence}.

\begin{figure}[!h]
    \centering
    \includegraphics[width=0.8\textwidth]{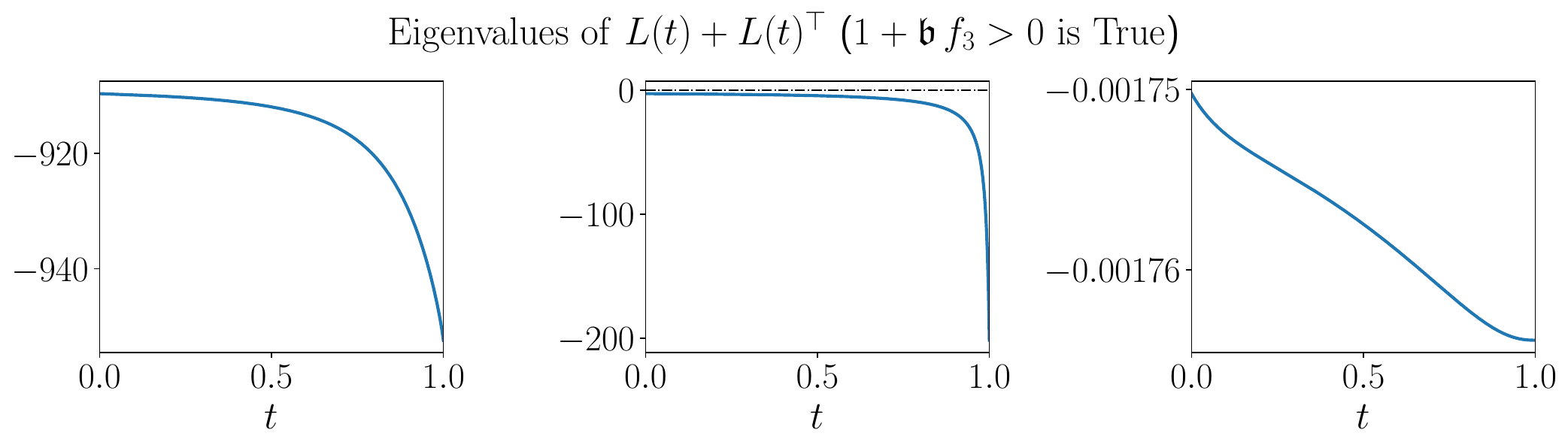}
    \caption{The first three leading eigenvalues of $L(t) + L(t)^\top$. The fourth eigenvalue is zero.}
    \label{fig:freiling_riccati_eigenvals}
\end{figure}

\subsection{Hamiltonian-Jacobi-Bellman equation for the alternative filter}
\label{calculations for hjb with alternative filter}
In the alternative filter, we have $\gamma_t = \nuIstar_t - f_3(t)\, \QI_t$; the broker observes an additional process $\gamma$. We use similar criterion in \eqref{eq: broker control problem} by replacing $\varB$ with $\varBalt$ and taking $\stateB = (s, x^B, q^B, q^I, \gamma, \hat{\alpha}, \nuU)$. After taking the ansatz $H(t, \stateB) = x^B + q^B\, s + h(t, \tilde{\stateB})$ where $\tilde{\stateB} = (q^B, q^I, \gamma, \hat{\alpha}, \nuU)$, we have the following HJB equation,
\begin{equation}
\label{eq: hjb equation alternative filter}
    \begin{split}
        0 = \partial_t  h + \sup_\nu \Bigg\{ & -\tempB\, \nuB^2 + \tempI\, \gamma^2 + 2\, \tempI\, f_3\, q^I\, \gamma + \tempI\, f_3^2\, (q^I)^2 + \tempU\, \nuU^2 + \permB\, \nuB\, q^B + q^B\, \hat{\alpha} - (\runcostBalt)\, (q^B)^2\\
        & + \partial_{q^B} h\, (\nuB - \gamma - f_3\, q^I - \nuU) + \partial_{q^I} h \, (\gamma + f_3\, q^I)\\
        & + \partial_{\gamma} h\, \left[ \mathfrak{G}^\alpha\, \hat{\alpha} + \mathfrak{G}_0\, \gamma + \mathfrak{G}_1\, \nuB + \mathfrak{G}_2 \right] + \frac{1}{2}\, \partial_{\gamma \gamma} h\, \mathfrak{G}_5^2 + \partial_{\gamma \hat{\alpha}} h\, \mathfrak{G}_5\, \left[ \mathfrak{G}_7\, \varBalt + \sigma^\alpha\, \mathfrak{K}\right]\\
        & -  \kappa^\alpha\, \partial_{\hat{\alpha}} h\, \hat{\alpha} + \frac{1}{2}\, \partial_{\hat{\alpha} \hat{\alpha}} h\, \left[ \mathfrak{G}_7\, \varBalt + \sigma^\alpha\, \mathfrak{K} \right]^2 - \kappa^U\, \partial_{\nuU} h\, \nuU + \frac{1}{2} (\sigma^U)^2\, \partial_{\nuU \nuU} h \Bigg\}\,
    \end{split}
\end{equation}
with $h(T, \tilde \stateB) = -(\termcostBalt_T)\, (q^B)^2$. One obtains
\[ \nuBstar = \frac{\partial_{q^B} h + \mathfrak{G}_1\, \partial_{\gamma} h + \permB\, q^B}{2\, \tempB}\,. \]
We take the ansatz $$h(t, \tilde{\stateB}) = h_{0,0}(t) + 2\, H_{0,1}(t)\, \tilde \eta^
\top + \tilde \eta\, H_2(t)\, \tilde \eta^\top\,,$$ where $h_{i,j}: [0,T] \to \R$ is a function for all $0\leq i \leq j \leq 5$, and
\begin{align}
    H_{0,1}(t) &:= \begin{pmatrix}
        h_{0,1}(t) & h_{0,2}(t) & h_{0,3}(t) & h_{0,4}(t) & h_{0,5}(t)
    \end{pmatrix}\\
    H_{1,1}(t) &:= \begin{pmatrix}
        h_{1,1}(t) & h_{1,2}(t) & h_{1,3}(t)\\
        h_{1,2}(t) & h_{2,2}(t) & h_{2,3}(t)\\
        h_{1,3}(t) & h_{2,3}(t) & h_{3,3}(t)
    \end{pmatrix}\,,\\
    H_{1,4}(t) &:= \begin{pmatrix}
        h_{1,4}(t) & h_{2,4}(t) & h_{3,4}(t)
    \end{pmatrix}^\top\,,\\
    H_{1,5}(t) &:= \begin{pmatrix}
        h_{1,5}(t) & h_{2,5}(t) & h_{3,5}(t)
    \end{pmatrix}^\top\,,\\
    H_{4,4}(t) &:= \begin{pmatrix}
        h_{4,4}(t) & h_{4,5}(t)\\
        h_{4,5}(t) & h_{5,5}(t)
    \end{pmatrix}\,,\\
    H_2(t) &:= \begin{pmatrix}
        H_{1,1}(t) & \begin{matrix} H_{1,4}(t) & H_{1,5}(t)\end{matrix}\\
        \begin{matrix} H_{1,4}(t)^\top \\ H_{1,5}(t)^\top \end{matrix} & H_{4,4}(t)
    \end{pmatrix}\,.
\end{align}
Substituting $\nuBstar$ and the ansatz, we obtain the following HJB equation in matrix form.
\begin{align}
    0 = \qquad \begin{pmatrix}
        q^B & q^I & \gamma
    \end{pmatrix}
    &\Bigg[ V_0(t) + V_1(t)\, H_{1,1}(t) + H_{1,1}(t)\, V_1(t)^\top + H_{1,1}(t)\, V_2\,(t)\, H_{1,1}(t) + H_{1,1}^{'}(t) \Bigg]\,
    \begin{pmatrix}
        q^B \\ q^I \\ \gamma
    \end{pmatrix}\\
    \quad +\, 2\, \hat \alpha\, \begin{pmatrix}
        q^B & q^I & \gamma
    \end{pmatrix}\, &\Bigg[ V_3(t) + V_4(t)\, H_{1,4}(t) + H_{1,4}'(t) \Bigg]\\
    \quad +\, 2\, \nuU\, \begin{pmatrix}
        q^B & q^I & \gamma
    \end{pmatrix}\, &\Bigg[ V_5(t) + V_6(t)\, H_{1,5}(t) + H_{1,5}'(t) \Bigg]\\
    \quad +\, \begin{pmatrix}
        \hat\alpha & \nuU
    \end{pmatrix}\, &\Bigg[ V_7(t) + V_8(t)\, H_{4,4}(t) + H_{4,4}(t)\, V_8(t)^\top +  H_{4,4}'(t) \Bigg]\, \begin{pmatrix}
        \hat\alpha \\ \nuU
    \end{pmatrix}\\
    \quad +\, 2\, \tilde \stateB \, &\Bigg[ U_0(t) + U_1(t)\, H_{0,1}(t)^\top + H_{0,1}'(t)^\top \Bigg]\\
    \quad +\,  &\Bigg[ U_2(t) + h_{0,0}'(t)  \Bigg]\,,
\end{align}
where
\begin{align}
    V_0(t) &= \begin{pmatrix}
    \frac{{\permB}^2}{4 \,{\tempB}}-{\rhoB}-{\rhooB}\, {\varBalt}(t) & 0 & 0 \\
    0 & {\tempI}\, {f_3}(t)^2 & {\tempI}\, {f_3}(t) \\
    0 & {\tempI}\, {f_3}(t) & {\tempI}
    \end{pmatrix}\,,\\
    V_1(t) &= \begin{pmatrix}
    \frac{{\permB}}{2\, {\tempB}} & 0 & \frac{{\permB}\, {\GF_1}(t)}{2 {\tempB}} \\
    -{f_3}(t) & {f_3}(t) & 0 \\
    -1 & 1 & {\GF_0}(t) 
    \end{pmatrix}\,,\\
    V_2(t) &= \begin{pmatrix}
    \frac{1}{{\tempB}} & 0 & \frac{{\GF_1}(t)}{{\tempB}} \\
    0 & 0 & 0 \\
    \frac{{\GF_1}(t)}{{\tempB}} & 0 & \frac{{\GF_1}(t)^2}{{\tempB}}
    \end{pmatrix}\,,\\
    V_3(t) &= \begin{pmatrix}
        \frac{1}{2} + \GF^\alpha(t)\, h_{1,3}(t) & \GF^\alpha(t)\, h_{2,3}(t) & \GF^\alpha(t)\, h_{3,3}(t)
    \end{pmatrix}\,,\\
    V_4(t) &=
    \begin{pmatrix}
     \frac{2 {\GF_1}(t) {h_{1,3}}(t)+2 {h_{1,1}}(t)-2 {\kappa^\alpha} {\tempB}+{\permB}}{2 {\tempB}} & 0 & \frac{{\GF_1}(t) (2 {\GF_1}(t) {h_{1,3}}(t)+2 {h_{1,1}}(t)+{\permB})}{2 {\tempB}} \\
     \frac{-{\tempB} {f_3}(t)+{\GF_1}(t) {h_{2,3}}(t)+{h_{1,2}}(t)}{{\tempB}} & {f_3}(t)-{\kappa^\alpha} & \frac{{\GF_1}(t) ({\GF_1}(t) {h_{2,3}}(t)+{h_{1,2}}(t))}{{\tempB}} \\
     \frac{{\GF_1}(t) {h_{3,3}}(t)+{h_{1,3}}(t)-{\tempB}}{{\tempB}} & 1 & \frac{{\tempB} {\GF_0}(t)+{\GF_1}(t) {h_{1,3}}(t)+{\GF_1}(t)^2 {h_{3,3}}(t)-{\kappa^\alpha} {\tempB}}{{\tempB}}
    \end{pmatrix}\,,\\
    V_5(t) &= -\begin{pmatrix}
        h_{1,1}(t) & h_{1,2}(t) & h_{1,3}(t)
    \end{pmatrix}^\top\,,\\
    V_6(t) &=
    \begin{pmatrix}
     \frac{2 {\GF_1}(t) {h_{1,3}}(t)+2 {h_{1,1}}(t)-2 {\kappa^U} {\tempB}+{\permB}}{2 {\tempB}} & 0 & \frac{{\GF_1}(t) (2 {\GF_1}(t) {h_{1,3}}(t)+2 {h_{1,1}}(t)+{\permB})}{2 {\tempB}} \\
     \frac{-{\tempB} {f_3}(t)+{\GF_1}(t) {h_{2,3}}(t)+{h_{1,2}}(t)}{{\tempB}} & {f_3}(t)-{\kappa^U} & \frac{{\GF_1}(t) ({\GF_1}(t) {h_{2,3}}(t)+{h_{1,2}}(t))}{{\tempB}} \\
     \frac{{\GF_1}(t) {h_{3,3}}(t)+{h_{1,3}}(t)-{\tempB}}{{\tempB}} & 1 & \frac{{\tempB} {\GF_0}(t)+{\GF_1}(t) {h_{1,3}}(t)+{\GF_1}(t)^2 {h_{3,3}}(t)-{\kappa^U} {\tempB}}{{\tempB}}
    \end{pmatrix}\,,\\
    V_7(t) &= \frac{1}{{\tempB}}\,
    \begin{pmatrix}
     \, \boxed{\begin{matrix} 2 {\GF_1}(t) {h_{1,4}}(t) {h_{3,4}}(t)+{\GF_1}(t)^2 {h_{3,4}}(t)^2\\+2 {\tempB} {\GF^\alpha}(t) {h_{3,4}}(t)+{h_{1,4}}(t)^2 \end{matrix}} & \boxed{\begin{matrix} {h_{1,4}}(t) ({\GF_1}(t) {h_{3,5}}(t)+{h_{1,5}}(t)-{\tempB})\\+{\GF_1}(t) {h_{1,5}}(t) {h_{3,4}}(t)\\ +{\GF_1}(t)^2 {h_{3,4}}(t) {h_{3,5}}(t)+{\tempB} {\GF^\alpha}(t) {h_{3,5}}(t) \end{matrix}}\, \\
     \, \boxed{\begin{matrix} {h_{1,4}}(t) ({\GF_1}(t) {h_{3,5}}(t)+{h_{1,5}}(t)-{\tempB})\\+{\GF_1}(t) {h_{1,5}}(t) {h_{3,4}}(t)\\+{\GF_1}(t)^2 {h_{3,4}}(t) {h_{3,5}}(t)+{\tempB} {\GF^\alpha}(t) {h_{3,5}}(t) \end{matrix}} & \boxed{\begin{matrix} 2 {\GF_1}(t) {h_{1,5}}(t) {h_{3,5}}(t)+{\GF_1}(t)^2 {h_{3,5}}(t)^2\\-2 {\tempB} {h_{1,5}}(t)+{h_{1,5}}(t)^2+{\tempB} {\tempU} \end{matrix}}
    \end{pmatrix}\,,\\
    V_8(t) &= \begin{pmatrix}
        -\kappa^\alpha & 0\\
        0 & -\kappa^U
    \end{pmatrix}\,,\\
    U_0(t) &= \begin{pmatrix}
        h_{1,3}(t) & h_{2,3}(t) & h_{3,3}(t) & h_{3,4}(t) & h_{3,5}(t)
    \end{pmatrix}\, \GF_2(t)\,,\\
    U_1(t) &=
    \begin{pmatrix}
     \frac{2 {\GF_1}(t) {h_{1,3}}(t)+2 {h_{1,1}}(t)+{\permB}}{2 {\tempB}} & 0 & \frac{{\GF_1}(t) (2 {\GF_1}(t) {h_{1,3}}(t)+2 {h_{1,1}}(t)+{\permB})}{2 {\tempB}} & 0 & 0 \\
     \frac{-{\tempB} {f_3}(t)+{\GF_1}(t) {h_{2,3}}(t)+{h_{1,2}}(t)}{{\tempB}} & {f_3}(t) & \frac{{\GF_1}(t) ({\GF_1}(t) {h_{2,3}}(t)+{h_{1,2}}(t))}{{\tempB}} & 0 & 0 \\
     \frac{{\GF_1}(t) {h_{3,3}}(t)+{h_{1,3}}(t)-{\tempB}}{{\tempB}} & 1 & {\GF_0}(t)+\frac{{\GF_1}(t) ({\GF_1}(t) {h_{3,3}}(t)+{h_{1,3}}(t))}{{\tempB}} & 0 & 0 \\
     \frac{{\GF_1}(t) {h_{3,4}}(t)+{h_{1,4}}(t)}{{\tempB}} & 0 & \frac{{\GF_1}(t) ({\GF_1}(t) {h_{3,4}}(t)+{h_{1,4}}(t))}{{\tempB}}+{\GF_\alpha}(t) & -{\kappa^\alpha} & 0 \\
     \frac{{\GF_1}(t) {h_{3,5}}(t)+{h_{1,5}}(t)-{\tempB}}{{\tempB}} & 0 & \frac{{\GF_1}(t) ({\GF_1}(t) {h_{3,5}}(t)+{h_{1,5}}(t))}{{\tempB}} & 0 & -{\kappa^U} \\
    \end{pmatrix}\,,\\
    U_2(t) &= -\frac{1}{{\tempB}}\, \Bigg[2\, {\GF_1}(t)\, {h_{0,1}}(t)\, {h_{0,3}}(t)+{\GF_1}(t)^2\, {h_{0,3}}(t)^2+2\, {\tempB}\, {\GF_2}(t)\, {h_{0,3}}(t)\\
    &\quad \quad \quad +2 \,{\tempB}\, {\GF_5}(t)\, {\GF_7}(t)\, {h_{3,4}}(t)\, {\varBalt}(t)+{\tempB}\, {\GF_5}(t)^2\, {h_{3,3}}(t)+2\, {\sigma^\alpha}\, {\tempB}\, {\GF_5}(t)\, {h_{3,4}}(t)\, \kF (t)\\
    &\quad \quad \quad +2\, {\sigma^\alpha}\, {\tempB}\, {\GF_7}(t)\, {h_{4,4}}(t)\, \kF (t)\, {\varBalt}(t)+{\tempB}\, {\GF_7}(t)^2\, {h_{4,4}}(t)\, {\varBalt}(t)^2+{h_{0,1}}(t)^2\\
    &\quad \quad \quad +(\sigma^\alpha)^2\, {\tempB}\, {h_{4,4}}(t)\, \kF (t)^2+(\sigma^U)^2\, {\tempB}\, {h_{5,5}}(t)\Bigg]\,.
\end{align}
Solving the HJB equation is equivalent to solving the following system of matrix differential equations
\begin{align}
    0 &= V_0(t) + V_1(t)\, H_{1,1}(t) + H_{1,1}(t)\, V_1(t)^\top + H_{1,1}(t)\, V_2\,(t)\, H_{1,1}(t) + H_{1,1}^{'}(t)\,, \label{eq: matrix riccati hjb alternative filter}\\
    0 &= V_3(t) + V_4(t)\, H_{1,4}(t) + H_{1,4}'(t)\,,\\
    0 &= V_5(t) + V_6(t)\, H_{1,5}(t) + H_{1,5}'(t)\,,\\
    0 &= V_7(t) + V_8(t)\, H_{4,4}(t) + H_{4,4}(t)\, V_8(t)^\top +  H_{4,4}'(t)\,,\\
    0 &= U_0(t) + U_1(t)\, H_{0,1}(t)^\top + H_{0,1}'(t)^\top\,,\\
    0 &= U_2(t) + h_{0,0}'(t)\,,
\end{align}
with terminal condition
\begin{align}
    H_{1,1}(T) &= \begin{pmatrix}
        -(\termcostBalt_T) & 0 &0\\
        0 & 0 & 0\\
        0 &0 & 0
    \end{pmatrix} \label{eq: matrix riccati terminal hjb alternative}
\end{align}
and all of the other terminal conditions are zero. One solves the individual matrix ODEs in the following order:
\begin{equation}
      H_{1,1} \to (H_{1,4}, H_{1,5}) \to H_{4,4} \to H_{0,1} \to h_{0,0}\,.
\end{equation}
Observe that the existence of the solution only depends on the first matrix equation \eqref{eq: matrix riccati hjb alternative filter} as the rest of the equations turn into linear ODEs when one obtains $H_{1,1}$.

\subsection{Second order effects}\label{app: second order}
Here, we explore how to incorporate this second-order effect into the informed trader's problem and understand its economic value.

\subsubsection{Informed trader}

From the perspective of the informed trader, a second-order effect can be incorporated by allowing the model for the broker's trading speed to be of the form
\begin{equation}
    \nuB_t = c_t\, \nuI_t + \underline{\nuB}_t\,,
\end{equation}
where $c_t$ is a deterministic function that reflects how the broker externalises his trading flow, and $\underline{\nuB}_t$ is the remainder of the broker's trading flow. In this setup, the informed trader assumes that $\underline{\nuB}$ is  exogenous and mean reverting, as in the initial model. Given this assumption, and similar to the approach above, the informed trader can filter $\underline{\nu_B}$ to compute  $\hat{\underline{\nuB}}$, and then  formulate an adaptive control problem. This control problem would be similar to that in \eqref{eq:trader_hjb}, but where one changes $\hat{\nuB}$ to $\underline{\hat{\nuB}}$, and adds a term of the form $\permB\, c_t\, \nuI_t\, \partial_s H^I$.  While we omit the detailed calculations here, the intuition is that the informed trader benefits in this modified model, as he effectively ``observes'' more of the broker's trading speed, which enhances his strategic advantage.

Another approach is for the informed trader to assume that the broker's trading follows the model
\begin{equation}
    \nuB_t = c_t^1\, \hat{\alpha}_t + c_t^2\, \QI_t + \underline{\nuB}_t\,,
\end{equation}
where $\hat{\alpha}_t$ is the broker's estimate of the signal, which the informed trader cannot construct independently.
Similar to the broker's problem in Section \ref{sec: broker}, the informed trader could substitute $\hat{\alpha}$  with her true signal and filter $\underline{\nuB}_t$, assuming the remaining trading flow of the broker is mean reverting and exogenous. Mathematically,  the price process would follow the dynamics
\begin{equation}
    \d S_t = \left[\permB\, \hat{\underline{\nuB}}_t + (1+\permB\, c_t^1)\, \alpha_t + \permB\, c_t^2\, \QI_t\right]\, \d t + \sigma^S\, \d \tilde{W}^{S,I}_t\,,
\end{equation}
where $\hat{\underline{\nuB}}_t$ is the filter of ${\underline{\nuB}}_t$. The techniques to solve such a control
problem are similar to what we report above, and thus omitted.

\subsubsection{Broker}

To study the impact of the second-order effects in the  strategy of the broker, we introduce the following modification to the initial broker's assumption \eqref{eq: broker_model_assumption} regarding the informed trader's strategy:
\begin{equation}
\label{eq: broker_belief}
    \nuIstar_t = f_1(t)\,\hat{\alpha}_t  + c\,f_2(t)\, \nuB_t + f_3(t)\,\QI[I*]_t\,,
\end{equation}
where the new parameter $c$ is employed to assess the degree to which the broker believes she can influence the behaviour of the informed trader.

Based on the model assumption \eqref{eq: broker_belief}, the broker computes her optimal speed $\nuB^{*, c}$. Note that the case $c=1$ corresponds to the initial model described in Section \ref{sec: broker}. Our aim is now to examine how $\nuB^{*, c}$ deviates from the optimal strategy in the original model where $c = 1$. To study the manipulation effects, we reduce the influence of the risk terms on the strategies of both the broker and the informed trader by setting $\beta_0^I = \beta_0^B = 10^{-5}$. Figure \ref{fig:second_order} reports our results.

\begin{figure}[H]
    \centering
    \includegraphics[width=0.8\textwidth]{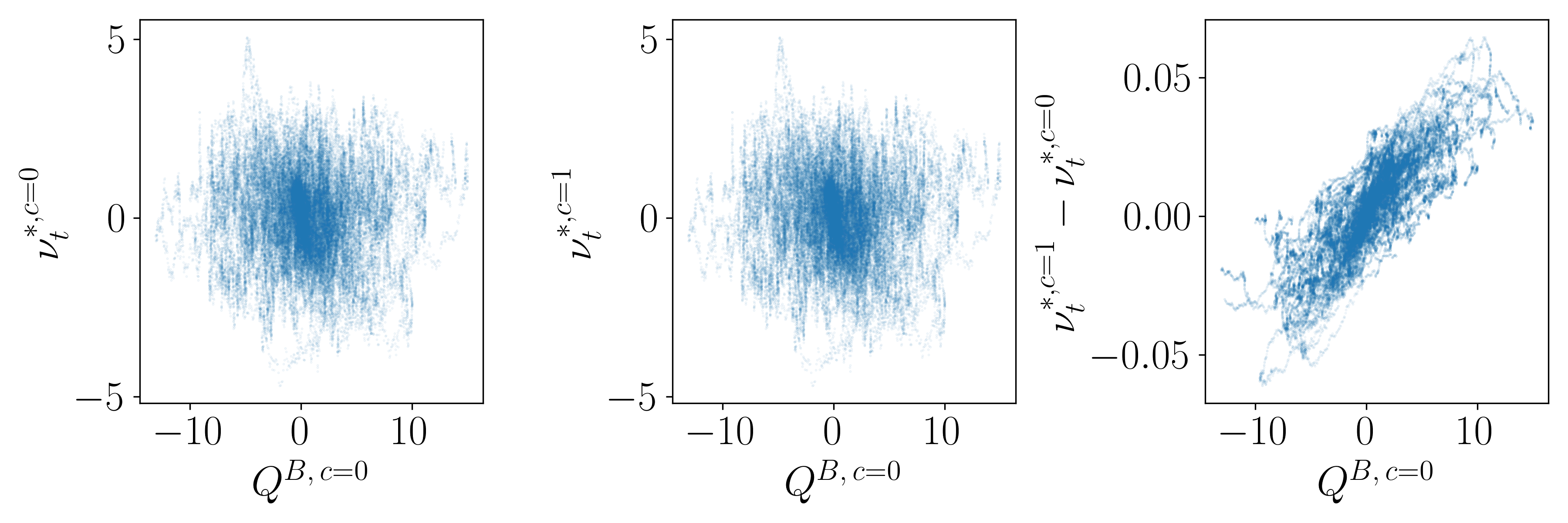}
    
    \caption{
    Scatter plot of $\nuB^{*,c=0}$, $\nuB^{*,c=1}$, and $\nuB^{*,c=1}-\nuB^{*,c=0}$ against the inventory of the broker for $c=0$ and $c=1$.
    \label{fig:second_order}
    }
\end{figure}

The broker's trading speed is a component in the informed trader's strategy \eqref{eq: broker_belief}. Thus, in our original model, the assumption \eqref{eq: broker_belief} may be interpreted as a certain ``belief'' of the broker in her ability to influence the informed trader's strategy. Figure \ref{fig:second_order} shows that the broker uses this component to her advantage. More precisely, the broker will, when  profitable, intentionally mislead the informed trader into trading in the opposite direction of her inventory. The aim is to offload her inventory to her client rather than unwinding  it  in the lit market herself. 

\end{document}